\documentclass[onefignum,onetabnum]{siamart171218}

\usepackage{amsmath}
\usepackage{enumerate}
\usepackage{mathrsfs}
\usepackage{color}
\usepackage{tikz}
\usepackage{mathtools}
\usepackage{breqn}
\usepackage{colortbl}
\usepackage{subfigure}

\def\Q{{\mathbb{Q}}}
\def\R{{\mathbb{R}}}

\newcommand{\defword}[1]{\textcolor{blue}{\em #1}}

\newtheorem{example}[theorem]{Example}

\newtheorem{question}[theorem]{Question}

\usepackage{lipsum}
\usepackage{amsfonts}
\usepackage{graphicx}
\usepackage{epstopdf}
\usepackage{algorithmic}
\ifpdf
  \DeclareGraphicsExtensions{.eps,.pdf,.png,.jpg}
\else
  \DeclareGraphicsExtensions{.eps}
\fi

\newsiamremark{remark}{Remark}
\newsiamremark{hypothesis}{Hypothesis}
\crefname{hypothesis}{Hypothesis}{Hypotheses}
\newsiamthm{claim}{Claim}

\headers{Hopf Bifurcations of Zero-One Reaction Networks}{Xiaoxian Tang,  and Kaizhang Wang}

\title{Hopf Bifurcations of Reaction Networks\\ with Zero-One Stoichiometric Coefficients\thanks{Submitted to the editors DATE.
\funding{XT was funded by the NSFC12001029.}
}}

\author{Xiaoxian Tang\thanks{School of Mathematical Sciences, Beihang University, Beijing,  China
  (\email{xiaoxian@buaa.edu.cn}, \url{https://sites.google.com/site/rootclassification/}).}
\and Kaizhang Wang\thanks{School of Mathematical Sciences, Beihang University, Beijing,  China
  (\email{wangkz@buaa.edu.cn}).}
}

\usepackage{amsopn}
\DeclareMathOperator{\diag}{diag}

\ifpdf
\hypersetup{
  pdftitle={Hopf Bifurcations of Reaction Networks with Zero-One Stoichiometric Coefficients},
  pdfauthor={Xiaoxian Tang, and Kaizhang Wang}
}
\fi

\begin{document}

\maketitle
\begin{abstract}
For the reaction networks with zero-one stoichiometric coefficients (or simply zero-one networks), we prove that if a network admits a Hopf bifurcation, then the rank of the stoichiometric matrix is at least four. 
 As a corollary, we show that if a zero-one network admits a Hopf bifurcation, then it contains at least four species and five reactions.
As applications, we show that there exist rank-four subnetworks, which have the capacity for Hopf bifurcations/oscillations,  in two biologically significant networks: the MAPK cascades and the ERK network. We provide a computational tool for computing all four-species, five-reaction, zero-one networks that have the capacity for Hopf bifurcations. 
\end{abstract}



\section{Introduction}\label{sec:intro}
Many dynamical systems that arise from biochemical reaction networks have the capacity for  Hopf bifurcations. The following question has attracted wide attention recently.
\begin{question}\label{ques:small}
When taken with mass-action kinetics, which small reaction networks admit a Hopf bifurcation?
\end{question}
This question is motivated by the oscillation problem of biochemical reaction networks \cite{Conradi2019, convex2015, osciMAPK2008, ploscptbio2007}, which is linked to the 
transduction in cellular systems. Deciding the existence 
of oscillations (periodic orbits) is as important as deciding the existence of multiple positive steady states \cite{bista2001,nc2020}. 
In practice, one approach for finding oscillations is to determine if the network admits a Hopf bifurcation \cite{biHB2022, otto2017, ERK2019}.  We say a network has the capacity for a Hopf bifurcation, if a reduced Jacobian matrix of the steady-state system
admits a pair of purely imaginary eigenvalues, while all other eigenvalues remain with nonzero real parts, and as some rate constant varies, 
a single pair of complex-conjugate eigenvalues crosses the imaginary axis (see Definition \ref{def:kxHB}) \cite{Conradi2019}.
There are a set of well-known algebraic criteria for deciding the existence of Hopf bifurcations based on the Hurwitz determinants, e.g., the Yang's criterion (see Lemma \ref{lemma:Yang}) \cite{Weber2000, Yang2002}. Algebraic/Symbolic methods are also developed for tackling the Hopf bifurcation/oscillation problems for many important biochemical reaction networks, e.g., 
the double phosphorylation cycle \cite{Conradi2019} and the ERK network \cite{ERK2019}. 
However, using such Hurwitz-based criteria  is   computationally challenging since the sizes of Hurwitz determinants are usually huge \cite{Conradi2019,convex2015,ERK2019}.
One approach for studying the oscillation problem of large networks is to look at related smaller networks since it is known that the capacity of a small network for oscillations can be inherited from an enlarged one 
\cite{inherit2018,inherit2022}.
Recently, in \cite{biHB2022}, Banaji and Boros address the question of which small bimolecular networks admit Hopf bifurcations,   and they fully classify three-species, four-reaction, bimolecular networks according to whether they admit or forbid Hopf bifurcations. 

In this paper, 
we study Question \ref{ques:small} 
for  
the reaction networks with zero-one stoichiometric coefficients (or simply zero-one networks). {\bf Our main theorem is that 
if a network with zero-one stoichiometric coefficients admits a Hopf bifurcation, then the rank of the stoichiometric matrix should be at least four (Theorem \ref{thm:rank<3})}.
A direct corollary of the  main result is that 
the smallest zero-one network that admits a Hopf bifurcation has four species and five reactions
(see Corollary \ref{coro:existHB4s5m} and 
Example \ref{exp:s4m5}).
We implement a computational tool
in {\tt Mathematica} 
for computing all four-species, five-reaction, zero-one networks that have the capacity for Hopf bifurcations (the number of such networks is about thirty thousand up to a natural equivalence), see
https://github.com/AspirinW/All4s5mCRNs.git.
As applications of the main theorem, we analyze two biologically significant networks: the MAPK cascades and the ERK network.
For these two networks, we demonstrate that there exist  rank-four subnetworks that admit Hopf bifurcations,  
and our computations show that the capacity for Hopf bifurcations/oscillations can be inherited from the original large networks. 

We remark that according to 
\cite{biHB2022},  it is known that the smallest rank
for a bimolecular network to admit Hopf bifurcations is three. 
However, it is nontrivial to show that for the zero-one networks, the smallest rank for admitting a Hopf bifurcation is four (our main theorem). 
In the proof of the main theorem,
we analyze the Jacobian matrix by the approach of extreme rays. This idea is inspired by the convex parameters introduced in \cite{Conradi2019}.
By the fact that the positive steady states of the dynamic systems can be represented by nonnegative combinations of extreme rays, we first transform the Jacobian matrix by changing the variables. 
For a zero-one network, we prove that the transformed Jacobian matrix  has the following useful properties.
The diagonal elements of the transformed Jacobian matrix are either a zero polynomial or a sum of monomials with negative coefficients (Corollary \ref{coro:jkkis-}).
The second order principal minors of the transformed Jacobian matrix and a set of useful  polynomials $d_{i,j,k}$ (see the definition in \eqref{def:dijk}) related to the third order principal minors  are either zero polynomials or sums of monomials with positive coefficients (Corollary \ref{lemma:positivePM}).
By the above properties, we study the relationship between the first Hurwitz determinant and the second order principal minors of the transformed Jacobian matrix, and the relationship between the second Hurwitz determinant and the polynomial $d_{i,j,k}$ (Lemma \ref{lemma:positivedeth}).
Then, we prove the key lemma for the main theory: 
 a reduced Jacobian matrix of the steady-state system does not 
admit a pair of purely imaginary eigenvalues, 
when the rank of the stoichiometric matrix is two or three (Lemma \ref{lemma:h1>0h2>0}).

The rest of this paper is organized as follows.
In Section \ref{sec:back}, we review the basic definitions and notions for the mass-action reaction networks, the extreme rays, and the Hopf bifurcations.
In Section \ref{sec:result}, we formally present the main result: Theorem \ref{thm:rank<3}.
We also present   applications for  illustrating the existence of rank-four subnetworks
that admit Hopf bifurcations from two biologically significant networks. 
In Section \ref{sec:trans}, we transform the Jacobian matrix 
by the method of extreme rays. In Section \ref{sec:structure},
we study the structures of the transformed Jacobian matrix. 
In Section \ref{sec:proofmain}, based on the lemmas proved in the previous two sections,  we prove  Theorem \ref{thm:rank<3}.
Finally, we end this paper with some future directions inspired by Theorem \ref{thm:rank<3}, see Section \ref{sec:dis}.

\section{Background}\label{sec:back}
In Section \ref{sec:network}, we briefly recall the standard notions and definitions of reaction networks, see \cite{CFMW, tx2020} for more details.
In Section \ref{sec:rays}, we review the flux cones and  the extreme rays.
In Section \ref{sec:hopf}, we recall the definitions of Hopf bifurcations, and the criteria based on the Hurwitz matrices. 

\subsection{Chemical reaction networks}\label{sec:network}

A \defword{reaction network} $G$  (or \defword{network} for short) consists of a set of $s$ species $\{X_1, \ldots, X_s\}$ and a set of $m$ reactions:
\begin{align}\label{eq:network}
\alpha_{1j}X_1 +
 \dots +
\alpha_{sj}X_s
~ \xrightarrow{\kappa_j} ~
\beta_{1j}X_1 +
 \dots +
\beta_{sj}X_s,
 \;
    {\rm for}~
	j=1, \ldots, m,
\end{align}
where all $\alpha_{i j}$ and $\beta_{i j}$ are non-negative integers, and $\left(\alpha_{1 j}, \ldots, \alpha_{s j}\right) \neq\left(\beta_{1 j}, \ldots, \beta_{s j}\right)$.
We call all $\alpha_{ij}$ and $\beta_{ij}$ the \defword{stoichiometric coefficients}.
We call each $\kappa_j \in \mathbb R_{>0}$ a \defword{rate constant}. 
We call the $s\times m$ matrix with $(i, j)$-entry equal to $\alpha_{ij}$ the \defword{reactant matrix} of $G$, denoted by ${\mathcal Y}$.
We call the $s\times m$ matrix with $(i, j)$-entry equal to $\beta_{ij}-\alpha_{ij}$ the \defword{stoichiometric matrix} of $G$, denoted by ${\mathcal N}$.
We call the image of ${\mathcal N}$ the \defword{stoichiometric subspace}, denoted by $S$.

We denote by $x_1, \ldots, x_s$ the concentrations of the species $X_1,\ldots, X_s$, respectively.
Under the assumption of mass-action kinetics, we describe how these concentrations change in time by the following system of ODEs:
\begin{align}\label{sys:f}
\dot{x}~=~f(\kappa,x)~:=~{\mathcal N}v(\kappa,x)
~,
\end{align}
where 
\begin{align}\label{eq:v1...vm}
v(\kappa,x)=(v_1(\kappa,x),\ldots,v_m(\kappa,x))^\top,
\end{align}
and 
\begin{align}\label{eq:vj}
v_j(\kappa,x)=\kappa_{j} \prod_{i=1}^{s} x_{i}^{\alpha_{i j}}.
\end{align}

For any $\kappa^*\in \mathbb{R}_{> 0}^m$, a \defword{steady state} 
of \eqref{sys:f} is a concentration vector $x^* \in \mathbb{R}_{\geq 0}^s$ such that $f(\kappa^*,x^*)=0$.
If all coordinates of a steady state $x^*$ are strictly positive, i.e., $x^*\in \mathbb{R}_{> 0}^s$, then we call $x^*$ a \defword{positive steady state}.
Denote by ${\rm Jac}_f(\kappa,x)$ the Jacobian matrix of $f(\kappa,x)$ with respect to $x$.
A steady state $x^*$ is \defword{nondegenerate} if ${\rm im}\left( {\rm Jac}_f(\kappa^*,x^*)|_{S} \right) = S$.

\subsection{Flux cones and extreme rays}\label{sec:rays}
Given a matrix $M \in \R^{s \times m}$, the \defword{flux cone} of $M$ is defined as
\begin{align*}
F(M) := \{r \in {\mathbb R}_{\geq 0}^m:M r = \bf0\},
\end{align*}
where we denote by $\bf0$ the vector whose coordinates are all zeros.
For any $r \in F(M)$, we call $r$ a \defword{ray}, if $r\neq \bf0$, and for any $\lambda > 0$, $\lambda r \in F(M)$. 
For any two rays $r^{(1)},r^{(2)} \in F(M)$, we say $r^{(1)}$ and $r^{(2)}$ are \defword{equivalent} if there exists $\lambda>0$ such that $r^{(1)}=\lambda r^{(2)}$.
We call a ray $r \in F(M)$ an \defword{extreme ray}, if for any two rays $r^{(1)},r^{(2)} \in F(M)$ such that  $r^{(i)}$ ($i\in \{1,2\}$) is nonequivalent with $r$, we have $r \notin \{\lambda r^{(1)} + (1-\lambda)r^{(2)}:0<\lambda<1\}$.
For any flux cone, if equivalent rays are not considered, then the number of extreme rays is finite and the choice of extreme rays is unique.
Assume that $R^{(1)},\ldots,R^{(t)}$ are the extreme rays of $F(M)$, where $t$ denotes the number of extreme rays. Then, it is well-known that any $r \in F(M)$ can be written as a non-negative combination of extreme rays
\begin{align*}
r=\sum_{i=1}^{t} \lambda_{i} R^{(i)}, \;\text{where}\; \lambda_{i} \geq 0\; \text{for any}\; i\in \{1, \ldots, t\}.
\end{align*}



\subsection{Hopf bifurcations}\label{sec:hopf}
In this section, we first recall the classical definition of Hopf bifurcation, and then, we clarify how to apply the classical definition  to the reaction networks. We consider a system of ODEs  parameterized by a single parameter $\mu \in \R$: 
\begin{align}\label{sys:g}
\dot{x}=g(\mu,x),
\end{align}
where $x \in \mathbb{R}^n$, and $g(\mu,x)$ is a smooth function in $(\mu, x)$.
Denoted by ${\rm Jac}_g(\mu,x)$ the Jacobian matrix of $g(\mu,x)$ with respect to $x$.
For some fixed value $\mu_0  \in \R$, let $x^* \in \mathbb{R}_{\geq 0}^n$ be a steady state of the system \eqref{sys:g}, i.e., $g(\mu_0,x^*)=0$.
 If $\operatorname{det}({\rm Jac}_g(\mu_0,x^*)) \neq 0$, then by the Implicit Function Theorem,  there exists a smooth curve of steady states $x(\mu)$ around   $\mu_0$ (i.e., $g(\mu,x(\mu))=0$ for all $\mu$ close enough to $\mu_0$) with $x(\mu_0)=x^*$.
We say \defword{the system \eqref{sys:g} has a Hopf bifurcation} (respectively, \defword{simple Hopf bifurcation}) \defword{at $(\mu_0,x^*)$ with respect to $\mu$}, if ${\rm Jac}_g(\mu_0,x^*)$ has a single pair of purely imaginary eigenvalues, while all other eigenvalues remain with nonzero (respectively, negative) real parts, and as $\mu$ varies, a single pair of complex-conjugate eigenvalues of ${\rm Jac}_g(\mu,x(\mu))$ crosses the imaginary axis.


We review two useful results: a condition for admitting a pair of purely imaginary roots (see Lemma \ref{prop:polyimroots}) and Yang's criterion for detecting  simple Hopf bifurcations (see Lemma \ref{lemma:Yang}).
We start by introducing the Hurwitz matrices.

\begin{definition}\label{def:hurwitz}
Let $p(z)=a_{0} z^{n}+a_{1} z^{n-1}+\cdots+a_{n}$ be a univariate   polynomial in $\Q[z]$ with $a_{0} \neq 0$. For any positive integer $i\;(i\leq n)$, the $i$-th \defword{Hurwitz matrix} of $p(z)$ is the following $i \times i$ matrix
\begin{align*}
H_{i}=\left(\begin{array}{ccccccc}
a_{1} & a_{0} & 0 & 0 & 0 & \cdots & 0 \\
a_{3} & a_{2} & a_{1} & a_{0} & 0 & \cdots & 0 \\
\vdots & \vdots & \vdots & \vdots & \vdots & \ddots & \vdots \\
a_{2 i-1} & a_{2 i-2} & a_{2 i-3} & a_{2 i-4} & a_{2 i-5} & \cdots & a_{i}
\end{array}\right),
\end{align*}
where the $(k, \ell)$-th entry is $a_{2 k-\ell}$ for $0 \leq 2 k-\ell \leq n$, and 0 otherwise. 
\end{definition}

For any square matrix $M$, we denote by $\operatorname{det}(M)$ the determinant of $M$.
We denote by $I$ the identity matrix. In the following two lemmas, we make a convention that 
$\operatorname{det}(H_{n-k}):=1$ whenever $n-k\leq 0$.

\begin{lemma}[Theorem 3.5 in \cite{Weber2000}]
\label{prop:polyimroots}
Let $p(z)=a_{0} z^{n}+a_{1} z^{n-1}+\cdots+a_{n}$ be a univariate  polynomial in $\Q [z]$ with $a_{0}>0$. 
For any positive integer $i\;(i\leq n)$, let $H_i$ be the $i$-th Hurwitz matrix of $p(z)$. Then, $p(z)$ has a pair of purely imaginary roots and all other roots with nonzero real parts if and only if
\begin{align*}
\operatorname{det}(H_{n-1})=0, \;\;\;\text{and}\;\;\; a_n\operatorname{det}(H_{n-2})\operatorname{det}(H_{n-3})>0.
\end{align*}
\end{lemma}

\begin{lemma}[Yang’s criterion \cite{Yang2002}]
\label{lemma:Yang}
Consider the system \eqref{sys:g}. 
Denoted by ${\rm Jac}_g(\mu,x)$ the Jacobian matrix of $g(\mu,x)$ with respect to $x$. 
Given $\mu_0 \in \R$ and a corresponding steady state $x^* \in \mathbb{R}_{\geq 0}^n$ with $\operatorname{det}({\rm Jac}_g(\mu_0,x^*)) \neq 0$, let $x(\mu)$ be a curve of steady states around $\mu_0$ with $x(\mu_0)=x^*$.
We define 
\begin{align*}
p(\mu;z):=\operatorname{det}(z I-{\rm Jac}_g(\mu,x(\mu)))=z^{n}+a_{1}(\mu) z^{n-1}+\cdots+a_{n}(\mu).
\end{align*}
For any positive integer $i\;(i\leq n)$, let $H_i(\mu)$ be the $i$-th Hurwitz matrix of $p(\mu;z)$.
Then, the system \eqref{sys:g} has a simple Hopf bifurcation at $(\mu_0,x^*)$ with respect to $\mu$ if and only if the following conditions hold:
\begin{enumerate}[(i)]
\item $\operatorname{det}(H_{n-1}(\mu_0))=0$ and $a_n(\mu_0)>0$,
\item $\operatorname{det}(H_{1}(\mu_0))>0,\ldots,\operatorname{det}(H_{n-2}(\mu_0))>0$, and
\item $\left.\frac{d\left(\operatorname{det} (H_{n-1}(\mu))\right)}{d \mu}\right|_{\mu=\mu_{0}} \neq 0$.
\end{enumerate}
\end{lemma}

Let $G$ be a network with a stoichiometric matrix $\mathcal{N} \in \R^{s \times m}$.
Denote the system of ODEs by $\dot{x}=f(\kappa,x)$ as in \eqref{sys:f}. Let ${\rm Jac}_f(\kappa,x)$ be the Jacobian matrix of $f(\kappa,x)$ with respect to $x$.
Let $r$ be the rank of $\mathcal{N}$.
It is remarkable that when $r<s$, for any $\kappa^*\in {\mathbb R}_{> 0}^m$ and for any $x^*\in {\mathbb R}_{>0}^s$, 
${\rm Jac}_f(\kappa^*, x^*)$ is singular. So, for such a network $G$, we can not directly apply the above definitions or Lemma \ref{lemma:Yang} to ${\rm Jac}_f(\kappa,x)$ for precluding or detecting Hopf bifurcations. Below, we introduce the reduced Jacobian matrix, and we clarify what we mean by ``a network admitting Hopf bifurcations" (see Definition \ref{def:kxHB}).

Let $A \in \R^{s \times r}$ be a matrix whose columns are a basis of stoichiometric subspace $S$.
Since the columns of $A$ are linearly independent, there exists a matrix $B \in \R^{r \times s}$ such that $B A = I$.
We define a \defword{reduced Jacobian matrix} of $f(\kappa,x)$ with respect to $x$
\begin{align}\label{def:reduceJac}
{\rm Jac}^{\text{red}}_f(\kappa,x) := B {\rm Jac}_f(\kappa,x) A.
\end{align}
\begin{remark}
Note that the different choices of $A$ and $B$ may result in different reduced Jacobian matrices. However, any two reduced Jacobian matrices are similar matrices. So,  they share the same characteristic polynomial and   eigenvalues.
\end{remark}
\begin{remark}\cite{Carsten2019}
For any $\kappa^*\in {\mathbb R}_{> 0}^m$, suppose $x^*\in {\mathbb R}_{\geq 0}^s$ is a steady state. Then, the steady state $x^*$ is nondegenerate if and only if $\operatorname{det}({\rm Jac}^{\text{red}}_f(\kappa^*,x^*))\neq 0$.
\end{remark}

\begin{definition}\label{def:kxHB}
For any network $G$ \eqref{eq:network}, denote  by  $\dot{x}~=~f(\kappa,x)$ the system of ODEs, see  \eqref{sys:f}. 
Let 
${\rm Jac}^{\text{red}}_f(\kappa,x)$ be a  reduced Jacobian matrix of $f(\kappa,x)$ with respect to $x$.
Given a vector of rate constants $\kappa^* \in \R^m_{>0}$, suppose $x^* \in \mathbb{R}_{> 0}^s$ is a corresponding nondegenerate positive steady state.
By the Implicit Function Theorem, for any $\kappa_i$ ($1 \leq i \leq m$), there exists a curve $\kappa(\kappa_i)$ around $\kappa^*_i$ with $\kappa(\kappa_i^*)=\kappa^*$ and a corresponding curve of positive steady states  $x(\kappa_i)$ around $\kappa^*_i$ with $x(\kappa^*_i)=x^*$.
We say \defword{the network  has a Hopf bifurcation} (respectively, \defword{simple Hopf bifurcation}) \defword{at $(\kappa^*,x^*)$ with respect to $\kappa_i$}, if ${\rm Jac}^{\text{red}}_f(\kappa^*,x^*)$ has a single pair of purely imaginary eigenvalues, while all other eigenvalues remain with nonzero (respectively, negative) real parts, and as $\kappa_i$ varies, a single pair of complex-conjugate eigenvalues of ${\rm Jac}^{\text{red}}_f(\kappa(\kappa_i),x(\kappa_i))$ crosses the imaginary axis.
We say \defword{the network  admits a Hopf bifurcation} (respectively, \defword{simple Hopf bifurcation}) if there exist $\kappa^* \in \R^m_{>0}$ and a corresponding nondegenerate positive steady state $x^* \in \mathbb{R}_{> 0}^s$ such that the network  has a Hopf bifurcation (respectively, simple Hopf bifurcation) at $(\kappa^*,x^*)$ with respect to some $\kappa_i$.
\end{definition}

\begin{remark}\label{rmk:onenps}
Definition \ref{def:kxHB} implies that a necessary condition for a network to admit a Hopf bifurcation is to admit a nondegenerate positive steady state.
\end{remark}

\begin{remark}\label{rmk:red}
Suppose we have a network with a stoichiometric matrix $\mathcal{N} \in \R^{s \times m}$.
Let $r=\text{rank}(\mathcal{N})$.
If $r<s$, then we can write the characteristic polynomial of ${\rm Jac}_f(\kappa,x)$ as
\begin{align*}
p(\kappa,x;z):=\operatorname{det}(z I-{\rm Jac}_f(\kappa,x))=z^{s-r}\left( z^{r}+a_{1}(\kappa,x) z^{r-1}+\cdots+a_{r}(\kappa,x) \right).
\end{align*}
Notice that the size of ${\rm Jac}^{\text{red}}_f(\kappa,x)$ is $r \times r$.
Also, notice that for any $\kappa^* \in \R^m_{>0}$ and for any $x^* \in \R^s_{>0}$, ${\rm Jac}^{\text{red}}_f(\kappa^*,x^*)$ have the same nonzero eigenvalues with ${\rm Jac}_f(\kappa^*,x^*)$ \cite{Carsten2019}. So, the polynomial 
\begin{align}\label{eq:phikxz}
\phi(\kappa,x;z) := z^{r}+a_{1}(\kappa,x) z^{r-1}+\cdots+a_{r}(\kappa,x)
\end{align}
is the characteristic polynomial of ${\rm Jac}^{\text{red}}_f(\kappa,x)$. In Section \ref{sec:result}, we will apply Lemma \ref{lemma:Yang} to $\phi(\kappa,x;z)$  for  detecting simple Hopf bifurcations. In Section \ref{sec:proofmain} (the proof of the main theorem),  we will apply  Lemma \ref{prop:polyimroots} to $\phi(\kappa,x;z)$ for precluding Hopf bifurcations.
\end{remark}

\section{Main Result}\label{sec:result}

We say a network \eqref{eq:network} is a \defword{network with zero-one stoichiometric coefficients} (or simply a \defword{zero-one network}), if 
the coefficients $\alpha_{ij}$ and $\beta_{ij}$ in \eqref{eq:network} belong to $\{0,1\}$ for all $i=1, \ldots, s$ and $j=1, \ldots, m$.
 Notice that a monomolecular network (e.g., \cite{nico2017}) is a special zero-one network.

\begin{theorem}\label{thm:rank<3}
Let $G$ be a zero-one network with a stoichiometric matrix $\mathcal{N}$. 
If $G$ admits a Hopf bifurcation, then $\text{rank}(\mathcal{N}) \geq 4$.
\end{theorem}



\begin{corollary}\label{coro:existHB4s5m}
Let $G$ be a zero-one network with $s$ species and $m$ reactions.
If $G$ admits a Hopf bifurcation, then $s\geq 4$, and $m\geq 5$.
\end{corollary}
\begin{proof}
Let $\mathcal{N}$ be the stoichiometric matrix of $G$.
Recall the size of $\mathcal{N}$ is $s \times m$.
So, $\text{rank}(\mathcal{N}) \leq \min\{s,m\}$.
Then, by Theorem \ref{thm:rank<3}, if the network $G$ admits a Hopf bifurcation, then we have $s \geq 4$ and $m \geq 4$. 
Suppose $m = 4$.
Then, we have $\text{rank}(\mathcal{N}) \leq m=4$.
By Theorem \ref{thm:rank<3}, we have ${\rm rank}({\mathcal N})\geq 4$.
So, ${\rm rank}({\mathcal N}) = 4$.
Notice that the number of columns of $\mathcal{N}$ is 4. 
So, the columns of $\mathcal{N}$ are linearly independent.
We denote by $c^{(1)},\ldots,c^{(4)}$ the columns of $\mathcal{N}$.
 Denote the system of ODEs by $\dot{x}=\mathcal{N}v(\kappa,x)$ as in \eqref{sys:f}.
By Remark \ref{rmk:onenps}, there exist $\kappa^* \in \R^4_{>0}$ and a corresponding positive steady state $x^* \in \R^s_{>0}$.
Then, the following equality holds.
\begin{align*}
\mathcal{N} v(\kappa^*, x^*) = \sum_{i=1}^4 v_{i}(\kappa^*, x^*) c^{(i)} = {\bf 0}.
\end{align*}
By the fact that $c^{(1)},\ldots,c^{(4)}$ are linearly independent, we have $v_{i}(\kappa^*, x^*) = 0$ for any $i \in \{1,2,3,4\}$. This is a contradiction to the fact that $v(\kappa^*, x^*) \in \mathbb{R}^4_{> 0}$.
Therefore, if $G$ admits a Hopf bifurcation, then we have $s \geq 4$ and $m \geq 5$.
\end{proof}


\begin{example}\label{exp:s4m5}
 As we have mentioned in Section \ref{sec:intro}, we are able to compute
by {\tt Mathematica} 
 all four-species, five-reaction, zero-one networks that have the capacity for Hopf bifurcations. 
The following network is one of them. 
\begin{align}\label{network:example}
X_1+X_2+X_3 &\xrightarrow{\kappa_1} X_2+X_3 \notag \\
X_3 &\xrightarrow{\kappa_2} X_1+X_3+X_4 \notag \\
X_1 &\xrightarrow{\kappa_3} X_1+X_2 \notag \\
X_2+X_4 &\xrightarrow{\kappa_4} X_3 \notag \\
X_1+X_2+X_3+X_4 &\xrightarrow{\kappa_5} X_1+X_4.
\end{align}
Denote the system of ODEs  by $\dot{x}=f(\kappa,x)$.
Let
\begin{align*}
\kappa^*=(\frac{7(\sqrt{1961}-39)}{440},\;\frac{1}{2},\;\frac{7(\sqrt{1961}-39)}{110},\;1,\;\frac{7(\sqrt{1961}-39)}{440}),
\end{align*}
Pick a corresponding positive steady state
\begin{align*}
x^* = (\frac{220}{7(\sqrt{1961}-39)},\;1,\;2,\;1).
\end{align*}
Below, we show that the network \eqref{network:example} has a simple Hopf bifurcation at $(\kappa^*,x^*)$ with respect to $\kappa_4$.
Let ${\rm Jac}^{\rm red}_f(\kappa,x)$ be a reduced Jacobian matrix of $f(\kappa,x)$ with respect to $x$ (see \eqref{def:reduceJac}).
The eigenvalues of ${\rm Jac}^{\rm red}_f(\kappa^*,x^*)$ are approximately
\begin{align*}
-2.960,\;-0.708,\;\mathbf{0.400 }\mathrm{i},\;-\mathbf{0.400 }\mathrm{i}.
\end{align*}
As the value of $\kappa_4$ varies, in the neighborhood of $\kappa_4^*$, consider the following curve of  rate constants 
\begin{align}\label{eq:kk4}
\kappa(\kappa_4)=(\kappa^*_1,\kappa^*_2,\kappa^*_3,\kappa_4,\frac{7(\sqrt{1961}-39)}{440}\kappa_4)\approx(0.0841,\;0.5,\;0.336,\;\kappa_4,\;0.0841\kappa_4).
\end{align}
and the  corresponding curve of positive steady states 
\begin{align*}
\label{eq:xk4}
x(\kappa_4) = (x^*_1,x^*_2,x^*_3,\frac{1}{\kappa_4})\approx(5.948,\;1,\;2,\;\frac{1}{\kappa_4}).
\end{align*}
Notice that $\kappa^*_{4}=1$.
When $\kappa_4=0.5$, the eigenvalues of ${\rm Jac}^{\rm red}_f(\kappa(\kappa_4),x(\kappa_4))$ are approximately
\begin{align*}
-2.449,\;-0.636,\;\mathbf{-0.0415+0.326}\mathrm{i},\;\mathbf{-0.0415-0.326}\mathrm{i}.
\end{align*}
When $\kappa_4=1.5$, the eigenvalues of ${\rm Jac}^{\rm red}_f(\kappa(\kappa_4),x(\kappa_4))$ are approximately
\begin{align*}
-3.469,\;-0.742,\;\mathbf{0.0214+0.442}\mathrm{i},\;\mathbf{0.0214-0.442}\mathrm{i}.
\end{align*}
In this case, a nearby oscillation is generated (see Figure \ref{fig:15xt}).
\begin{figure}[h!]
\centering
\includegraphics[scale=0.5]{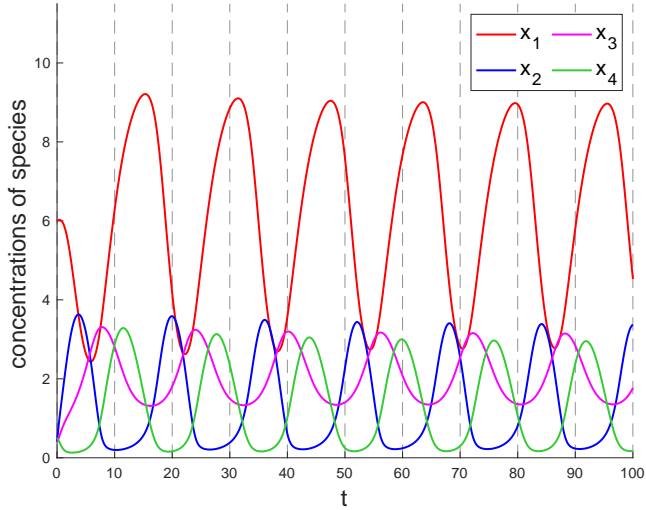}
\caption{The figure shows oscillations in the species of network \eqref{network:example}.
The rate constants are given by $\kappa(\kappa_4)$ in \eqref{eq:kk4} where the value of $\kappa_4$ is 1.5.
We set the initial concentrations to $x^{(0)}=(6,0.4,0.4,0.5)$.}
\label{fig:15xt}
\end{figure}

Note that in the neighborhood of $\kappa_4^*$,  the eigenvalues of ${\rm Jac}^{\rm red}_f(\kappa(\kappa_4),x(\kappa_4))$ change continuously as the value of $\kappa_4$ varies.
We observe that when the value of $\kappa_{4}$ changes from 0.5 to 1.5, the real parts of a pair of conjugate-complex eigenvalues change from negative to positive, and become zero when $\kappa_{4}=\kappa_{4}^*=1$.
Moreover, all other eigenvalues remain with negative real parts.
So, by Definition \ref{def:kxHB}, the network \eqref{network:example} has a simple Hopf bifurcation at $(\kappa^*,x^*)$ with respect to $\kappa_{4}$. Later in Example \ref{ex:details}, we will explain how to find $(\kappa^*,x^*)$ by Yang's criterion after we transform the Jacobian matrix by the method of extreme rays.  
\end{example}

\subsection{Application: MAPK cascades}
Theorem \ref{thm:rank<3} states that if a zero-one network admits a Hopf bifurcation, then the rank of the  stoichiometric matrix is at least four. 
In this subsection, we show that 
 a rank-four subnetwork that admits a Hopf bifurcation can be  obtained from  the well-known mitogen-activated-protein-kinase (MAPK) cascade,
 which represents a crucial step of the chemical signal transduction in cellular systems and is widely conserved in eukaryotes \cite{chang2001}. 
 In fact, this rank-four subnetwork is first found in \cite{otto2017}, where the authors do not emphasize it has the minimum rank that admits a Hopf bifurcation. Here,  we use the same notation with that used in \cite{otto2017}.
The MAPK cascade network,  denote by $\text{Net}_{\text{MAPK}}$, consists of 14 species and 18 reactions (see \eqref{3.7a} to \eqref{3.7f}).
{\footnotesize\begin{subequations}\label{network:MAPK}
\begin{align}
B+A^* &\xrightleftharpoons[\kappa_{13}]{\kappa_1} A^*B \xrightarrow{\kappa_2} A^*+B_1 \label{3.7a}\\
A^*+B_1 &\xrightleftharpoons[\kappa_{14}]{\kappa_3} A^*B_1 \xrightarrow{\kappa_4} A^*+B_2 \label{3.7b} \\
B_2+C &\xrightleftharpoons[\kappa_{15}]{\kappa_5} CB_2 \xrightarrow{\kappa_6} C+B_1 \label{3.7c} \\
B_1+C &\xrightleftharpoons[\kappa_{16}]{\kappa_{7}} CB_1 \xrightarrow{\kappa_{8}} C+B \label{3.7d} \\
D_1+A &\xrightleftharpoons[\kappa_{17}]{\kappa_{9}} D_1A \xrightarrow{\kappa_{10}} D_1+A^* \label{3.7e} \\
D_2+A^* &\xrightleftharpoons[\kappa_{18}]{\kappa_{11}} D_2A^* \xrightarrow{\kappa_{12}} D_2+A \label{3.7f}.
\end{align}
\end{subequations}
}The authors in \cite{otto2017} present a small oscillating network obtained from $\text{Net}_{\text{MAPK}}$, denoted by $\text{SubNet}_{\text{MAPK}}$, which consists of 6 species and 7 reactions  (see \eqref{3.8a} to \eqref{3.8f}).
{\footnotesize\begin{subequations}\label{network:subMAPK}
\begin{align}
B+A^* \xrightarrow{\tau_1} &A^*B \xrightarrow{\tau_2} A^*+B_1\label{3.8a} \\
A^*+B_1 &\xrightarrow{\tau_3} A^*+B_2 \label{3.8c} \\
B_2 &\xrightarrow{\tau_4} B_1\label{3.8d} \\
B_1 &\xrightarrow{\tau_5} B \label{3.8e} \\
A &\xrightleftharpoons[\tau_7]{\tau_6} A^* \label{3.8f}.
\end{align}
\end{subequations}
}First, we explain how to obtain $\text{SubNet}_{\text{MAPK}}$ from $\text{Net}_{\text{MAPK}}$.
\begin{enumerate}[(i)]
\item We obtain \eqref{3.8a} by   removing the reverse reaction indexed by the rate constant $\kappa_{13}$ from \eqref{3.7a}. 
\item We obtain \eqref{3.8c} by removing  the reverse reaction indexed by the rate constant $\kappa_{14}$ and the intermediate $A^*B_1$ from \eqref{3.7b}.
\item   We obtain \eqref{3.8d} by removing the reverse reaction indexed by the rate constant $\kappa_{15}$, the intermediate $CB_2$,  and the specie $C$ from \eqref{3.7c}.
Similarly, we obtain  \eqref{3.8e} from \eqref{3.7d}, and we 
obtain \eqref{3.8f} from \eqref{3.7e}--\eqref{3.7f}.
\end{enumerate}
It is straightforward to check that the rank of the stoichiometric matrix of the network $\text{SubNet}_{\text{MAPK}}$ is four.
Next, we show that $\text{SubNet}_{\text{MAPK}}$ admits a simple Hopf bifurcation according to Definition \ref{def:kxHB}.
We denote by $y_1,\ldots,y_{6}$ the concentrations of 
the $6$ species in $\text{SubNet}_{\text{MAPK}}$, see Table \ref{tab:subMAPK}.
{\small\begin{table}[ht]
\centering
\caption{Species concentrations of $\text{SubNet}_{\text{MAPK}}$}
\begin{tabular}{cccccc}
\hline$y_{1}$ & $y_{2}$ & $y_{3}$ & $y_{4}$ & $y_{5}$ & $y_{6}$ \\
\hline$A^*$ & $B$ & $B_1$ & $A^*B$ & $B_2$ & $A$ \\
\hline
\end{tabular}
\label{tab:subMAPK}
\end{table}
}Let $\tau=(\tau_1,\ldots,\tau_{7})$, and let $y=(y_1,\ldots,y_{6})$.
Denote the system of ODEs by $\dot{y}=\hat{f}(\tau,y)$.
Let ${\rm Jac}^{\rm red}_{\hat{f}}(\tau,y)$ be a reduced Jacobian matrix of $\hat{f}(\tau,y)$ with respect to $y$.
Let
\begin{align*}
\tau^*=(100,\;1,\;100,\;1,\;10,\;2.39,\;0.239).
\end{align*}
And, pick a corresponding positive steady state 
\begin{align*}
y^*=(0.1,\;0.1,\;0.1,\;1,\;1,\;1).
\end{align*}
Notice here that this point exactly lies on the Hopf bifurcation curve given in \cite[Figure 7]{otto2017}. 
The eigenvalues of ${\rm Jac}^{\rm red}_{\hat{f}}(\tau^*,y^*)$ are approximately
\begin{align*}
-22.314 + 6.993 \mathrm{i},\;-22.314 - 6.993 \mathrm{i},\;\textbf{0.815}\mathrm{i},\;\textbf{-0.815}\mathrm{i}.
\end{align*}
As the value of $\tau_7$ varies, in the neighborhood of $\tau_7^*$, consider the curve of  rate constants
\begin{align}\label{eq:tt7}
\tau(\tau_7) = (100,\;1,\;100,\;1,\;10,\;10\tau_7,\;\tau_7).
\end{align}
We remark that for every point on this curve, $y^*$ is a  corresponding positive steady state.
When $\tau_7=0.228$, the eigenvalues of ${\rm Jac}^{\rm red}_{\hat{f}}(\tau(\tau_7),y^*)$ are approximately
\begin{align*}
-22.282 + 6.995 \mathrm{i},\;-22.282 - 6.995 \mathrm{i},\;\textbf{0.0284+0.797}\mathrm{i},\;\textbf{0.0284-0.797}\mathrm{i}.
\end{align*}
In this case, a nearby oscillation is generated around the positive steady state $y^*$, see Figure \ref{fig:subMAPK}, where
we set the initial concentrations to 
\begin{align}\label{eq:x00}
y^{(0)}=(0.15,\;0.15,\;0.15,\;1.1,\;1.1,\;1.1).
\end{align} 
\begin{figure}[h!]
\centering
\includegraphics[scale=0.5]{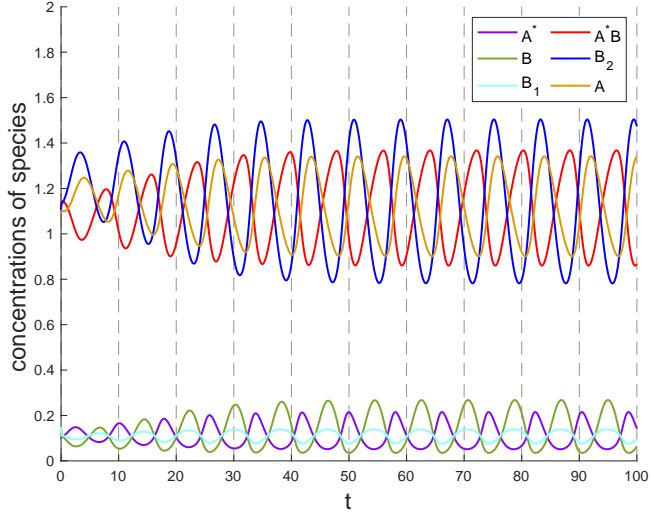}
\caption{The network $\text{SubNet}_{\text{MAPK}}$ gives rise to oscillations with respect to the species in Table \ref{tab:subMAPK}. 
The rate constants are given by $\tau(\tau_7)$ in \eqref{eq:tt7} where the value of $\tau_{7}$ is 0.228.
We set the initial concentrations to $y^{(0)}$ in \eqref{eq:x00}.}
\label{fig:subMAPK}
\end{figure}

\noindent When $\tau_7=0.248$, the eigenvalues of ${\rm Jac}^{\rm red}_{\hat{f}}(\tau(\tau_7),y^*)$ are approximately
\begin{align*}
-22.340 + 6.991 \mathrm{i},\;-22.340 - 6.991 \mathrm{i},\;\textbf{-0.0238+0.829}\mathrm{i},\;\textbf{-0.0238-0.829}\mathrm{i}.
\end{align*}
We observe that when the value of $\tau_7$ changes from 0.228 to 0.248, the real parts of a pair of conjugate-complex eigenvalues changes from positive to negative, and become zero when $\tau_7=\tau^*_7$.
Moreover, all other eigenvalues remain with negative real parts.
So, by Definition \ref{def:kxHB}, $\text{SubNet}_{\text{MAPK}}$ has a simple Hopf bifurcation at $(\tau^*,y^*)$ with respect to $\tau_7$.

Below, we show that the original network $\text{Net}_{\text{MAPK}}$ also admits a simple Hopf bifurcation.
Let $x_1,\ldots,x_{14}$ denote the concentrations of the species in $\text{Net}_{\text{MAPK}}$, see Table \ref{tab:MAPK}.
{\small\begin{table}[ht]
\centering
\caption{Species concentrations of $\text{Net}_{\text{MAPK}}$.}
\resizebox{\textwidth}{!}{
\begin{tabular}{cccccccccccccc}
\hline$x_{1}$ & $x_{2}$ & $x_{3}$ & $x_{4}$ & $x_{5}$ & $x_{6}$ & $x_{7}$ & $x_{8}$ & $x_{9}$ & $x_{10}$ & $x_{11}$ & $x_{12}$ & $x_{13}$ & $x_{14}$ \\
\hline$A^*$ & $B$ & $B_1$ & $A^*B$ & $B_2$ & $A$ & $D_1$ & $D_2$ & $D_1A$ & $D_2A^*$ & $C$ & $CB_1$ & $CB_2$ & $A^*B_1$ \\
\hline
\end{tabular}}
\label{tab:MAPK}
\end{table}
}Let $\kappa=(\kappa_1,\ldots,\kappa_{18})$, and let $x=(x_1,\ldots,x_{14})$.
Denote the system of ODEs  by $\dot{x}=f(\kappa,x)$.
Let ${\rm Jac}^{\rm red}_{f}(\kappa,x)$ be a reduced Jacobian matrix of $f(\kappa,x)$ with respect to $x$.
Let
{\small \begin{align*}
\kappa^*=(200,\;1,\;200,\;10,\;2,\;1,\;20,\;1,\;1.00579,\;0.00579,\;1.0579,\;0.00579,\;1,\;10,\;1,\;1,\;1,\;0.1).
\end{align*}}
Then, pick a positive steady state corresponding to  $\kappa^*$ 
\begin{align*}
x^*=(0.1,\;0.1,\;0.1,\;1,\;1,\;1,\;1,\;1,\;1,\;1,\;1,\;1,\;1,\;0.1).
\end{align*}
The eigenvalues of ${\rm Jac}^{\rm red}_f(\kappa^*,x^*)$ are approximately
\begin{align*}
-76.416,\;-32.662,\;-15.519,\;-5.290,\;-3.364,\;-3.015,\;\textbf{0.225}\mathrm{i},\;\textbf{-0.225}\mathrm{i}.
\end{align*}
As the value of $\kappa_{10}$ varies, in the neighborhood of $\kappa_{10}^*$, consider the curve of  rate constants 
{\small \begin{align}\label{eq:kk10}
\kappa(\kappa_{10}) = (200,\;1,\;200,\;10,\;2,\;1,\;20,\;1,\;1+\kappa_{10},\;\kappa_{10},\;1+10\kappa_{10},\;\kappa_{10},\;1,\;10,\;1,\;1,\;1,\;0.1),
\end{align}
}and  notice that for every point on the curve, $x^*$ is a positive steady state.
When $\kappa_{10}=0.00479$, the eigenvalues of ${\rm Jac}^{\rm red}_f(\kappa(\kappa_{10}),x^*)$ are approximately
{\small \begin{align*}
-76.411,-32.661,-15.519,-5.291,-3.363,-3.0128,\textbf{0.00103+0.223}\mathrm{i},\textbf{0.00103-0.223}\mathrm{i}.
\end{align*}
}In this case, a nearby oscillation is generated around the positive steady state $x^*$, see Figure \ref{fig:MAPK}, where we set the initial concentrations to 
\begin{align}\label{eq:x0}
x^{(0)}=(0.15,\;0.15,\;0.15,\;1.1,\;1.1,\;1.1,\;1.1,\;1.1,\;1.1,\;1.1,\;1.1,\;1.1,\;1.1,\;0.15).
\end{align} 
Notice that here,  we see that the capacity of $\text{SubNet}_{\text{MAPK}}$ for oscillations is inherited from  $\text{Net}_{\text{MAPK}}$ by the facts that $y^*$ exactly gives the first six coordinates of $x^*$, and $y^{(0)}$ exactly gives the first six coordinates of $x^{(0)}$. 
\begin{figure}[htbp]
\centering
\includegraphics[scale=0.5]{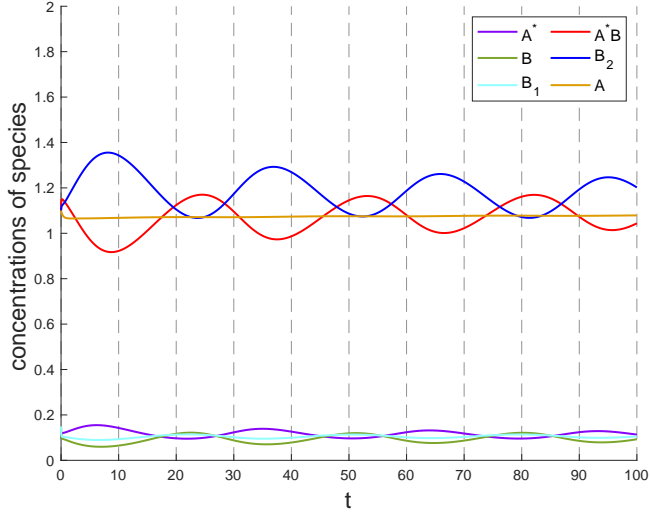}
\caption{The network $\text{Net}_{\text{MAPK}}$ gives rise to oscillations, and we show here only the first six species in Table \ref{tab:MAPK}. 
The rate constants are given by $\kappa(\kappa_{10})$ in \eqref{eq:kk10} where the value of $\kappa_{10}$ is 0.00479.
We set the initial concentrations to $x^{(0)}$ in \eqref{eq:x0}.
We remark that the maximum amplitude of the orange curve is about $0.01$.}
\label{fig:MAPK}
\end{figure}
When $\kappa_{10}=0.00679$, the eigenvalues of ${\rm Jac}^{\rm red}_f(\kappa(\kappa_{10}),x^*)$ are approximately
{\small \begin{align*}
-76.421,-32.664,-15.520,-5.290,-3.365,-3.0181,\textbf{-0.00104+0.227}\mathrm{i},\textbf{-0.00104-0.227}\mathrm{i}.
\end{align*}
}We observe that when the value of $\kappa_{10}$ changes from 0.00479 to 0.00679, the real parts of a pair of conjugate-complex eigenvalues change from positive to negative, and become zero when $\kappa_{10}=\kappa_{10}^*$.
Moreover, all other eigenvalues remain with negative real parts.
So, by Definition \ref{def:kxHB}, $\text{Net}_{\text{MAPK}}$ has a simple Hopf bifurcation at $(\kappa^*,x^*)$ with respect to $\kappa_{10}$.

\subsection{Application: ERK network}
In this subsection, we consider the phosphorylation and dephosphorylation of extracellular signal-regulated kinase (ERK). We will show the existence of a rank-four subnetwork in the ERK network that admits a Hopf bifurcation. 
The authors in \cite{ERK2019} introduce the reduced ERK network, and they  show that the reduced network preserves oscillations. As shown in \eqref{network:redERKa} to \eqref{network:redERKe}, the reduced ERK network, denoted by $\text{Net}_{\text{redERK}}$, contains 10 species and 10 reactions.
{\footnotesize\begin{subequations}\label{network:redERK}
\begin{align}
S_{00}+E \xrightarrow{\kappa_1} &S_{00}E \xrightarrow{\kappa_2} S_{01}E  \label{network:redERKa}\\
S_{01}E \xrightarrow{\kappa_3} S_{01}+E,&\; S_{01}+F\xrightarrow{\kappa_4} S_{00}+F \label{network:redERKb} \\
S_{11}+F \xrightarrow{\kappa_5} &S_{11}F \xrightarrow{\kappa_6} S_{10}F \label{network:redERKc}\\
S_{10}F \xrightarrow{\kappa_7} S_{10}+F,&\; S_{10}+E\xrightarrow{\kappa_8} S_{11}+E \label{network:redERKd} \\
S_{00}E \xrightarrow{\kappa_9} S_{11}+E,&\; S_{10}+F \xrightarrow{\kappa_{10}} S_{00}+F \label{network:redERKe}.
\end{align}
\end{subequations}
}Below, we give  a  smaller network, denoted by $\text{SubNet}_{\text{redERK}}$, containing  7 species and 7 reactions.
{\footnotesize\begin{subequations}\label{network:subredERK}
\begin{align}
S_{00}+E &\xrightarrow{\tau_1} S_{01}E \label{network:subredERKa}\\
S_{01}E + F &\xrightarrow{\tau_2} S_{00}+E+F \label{network:subredERKb} \\
S_{11}+F &\xrightarrow{\tau_3} S_{10}F \label{network:subredERKc}\\
S_{10}F \xrightarrow{\tau_4} S_{10}+F,&\; S_{10}+E\xrightarrow{\tau_5} S_{11}+E \label{network:subredERKd} \\
S_{00}E \xrightarrow{\tau_6} S_{11}+E,&\; S_{10}+F \xrightarrow{\tau_7} S_{00}+F \label{network:subredERKe}.
\end{align}
\end{subequations}
}First, we explain how to obtain $\text{SubNet}_{\text{redERK}}$ from $\text{Net}_{\text{redERK}}$.
\begin{enumerate}[(i)]
\item We obtain \eqref{network:subredERKa} by removing 
the intermediate $S_{00}E$ from 
\eqref{network:redERKa}. Similarly, we obtain \eqref{network:subredERKc} from \eqref{network:redERKc}.
\item In \eqref{network:redERKb}, we remove the specie $S_{01}$ and combine the reactions indexed by the rate constants $\kappa_{3}$ and $\kappa_{4}$ (i.e., combine the reactants and the products from the two reactions, respectively). Then, we obtain \eqref{network:subredERKb}.
\item The reactions \eqref{network:subredERKd} and \eqref{network:subredERKe} are identical to \eqref{network:redERKd} and \eqref{network:redERKe}.
\end{enumerate}
It is straightforward to check that the rank of the stoichiometric matrix of the network $\text{SubNet}_{\text{redERK}}$ is four.
Next, we show that $\text{SubNet}_{\text{redERK}}$ admits a simple Hopf bifurcation.
Let $y_1,\ldots,y_{6}$ be the concentrations of the  species in $\text{SubNet}_{\text{redERK}}$, see Table \ref{tab:subredERK}.
\begin{table}[ht]
\centering
\caption{Species concentrations of $\text{SubNet}_{\text{redERK}}$.}
\begin{tabular}{ccccccc}
\hline$y_{1}$ & $y_{2}$ & $y_{3}$ & $y_{4}$ & $y_{5}$ & $y_{6}$ & $y_{7}$ \\
\hline$S_{00}$ & $E$ & $S_{01}E$ & $S_{11}$ & $S_{10}$ & $F$ & $S_{10}F$ \\
\hline
\end{tabular}
\label{tab:subredERK}
\end{table}
Let $\tau=(\tau_1,\ldots,\tau_{7})$, and let $y=(y_1,\ldots,y_{7})$.
Denote the system of ODEs by $\dot{y}=\hat{f}(\tau,y)$.
Let ${\rm Jac}^{\rm red}_{\hat{f}}(\tau,y)$ be a reduced Jacobian matrix of $\hat{f}(\tau,y)$ with respect to $y$.
Let
\begin{align*}
\tau^*=(3.0682,\;20.682,\;110,\;0.1,\;0.1,\;1,\;1).
\end{align*}
Pick a positive steady state corresponding to $\tau^*$:
\begin{align*}
y^*=(1,\;1,\;1,\;0.1,\;1,\;0.1,\;1).
\end{align*}
We remark that we find the point $(\tau^*,y^*)$  by the condition (i) in Lemma \ref{lemma:Yang}.
The eigenvalues of ${\rm Jac}^{\rm red}_{\hat{f}}(\tau^*,y^*)$ are approximately
\begin{align*}
-21.814,\;-10.590,\;\textbf{0.0724}\mathrm{i},\;\textbf{-0.0724}\mathrm{i}.
\end{align*}
As the value of $\tau_1$ varies, in the neighborhood of $\tau_1^*$, consider the curve of  rate constants
\begin{align}\label{eq:tt1}
\tau(\tau_1) = (\tau_1,\;10\tau_1-10,\;110,\;0.1,\;0.1,\;1,\;1).
\end{align}
Notice that all points on this curve admit a common positive steady state $y^*$.
When $\tau_1=3.06$, the eigenvalues of ${\rm Jac}^{\rm red}_{\hat{f}}(\tau(\tau_1),y^*)$ are approximately
\begin{align*}
-21.819,\;-10.560,\;\textbf{-0.000568+0.0733}\mathrm{i},\;\textbf{-0.000568-0.0733}\mathrm{i}.
\end{align*}
When $\tau_1=3.08$, the eigenvalues of ${\rm Jac}^{\rm red}_{\hat{f}}(\tau(\tau_1),y^*)$ are approximately
\begin{align*}
-21.807,\;-10.634,\;\textbf{0.000819+0.0712}\mathrm{i},\;\textbf{0.000819-0.0712}\mathrm{i}.
\end{align*}
In this case, a nearby oscillation is generated around the positive steady state $y^*$, see figure \ref{fig:subredERK},
where we set the initial concentrations to 
\begin{align}\label{eq:y0}
y^{(0)}=(1,\;1,\;1,\;0.0999999,\;1,\;0.0999999,\;1).
\end{align} 
\begin{figure}[h!]
\centering
\subfigure{
\includegraphics[width=0.45\linewidth]{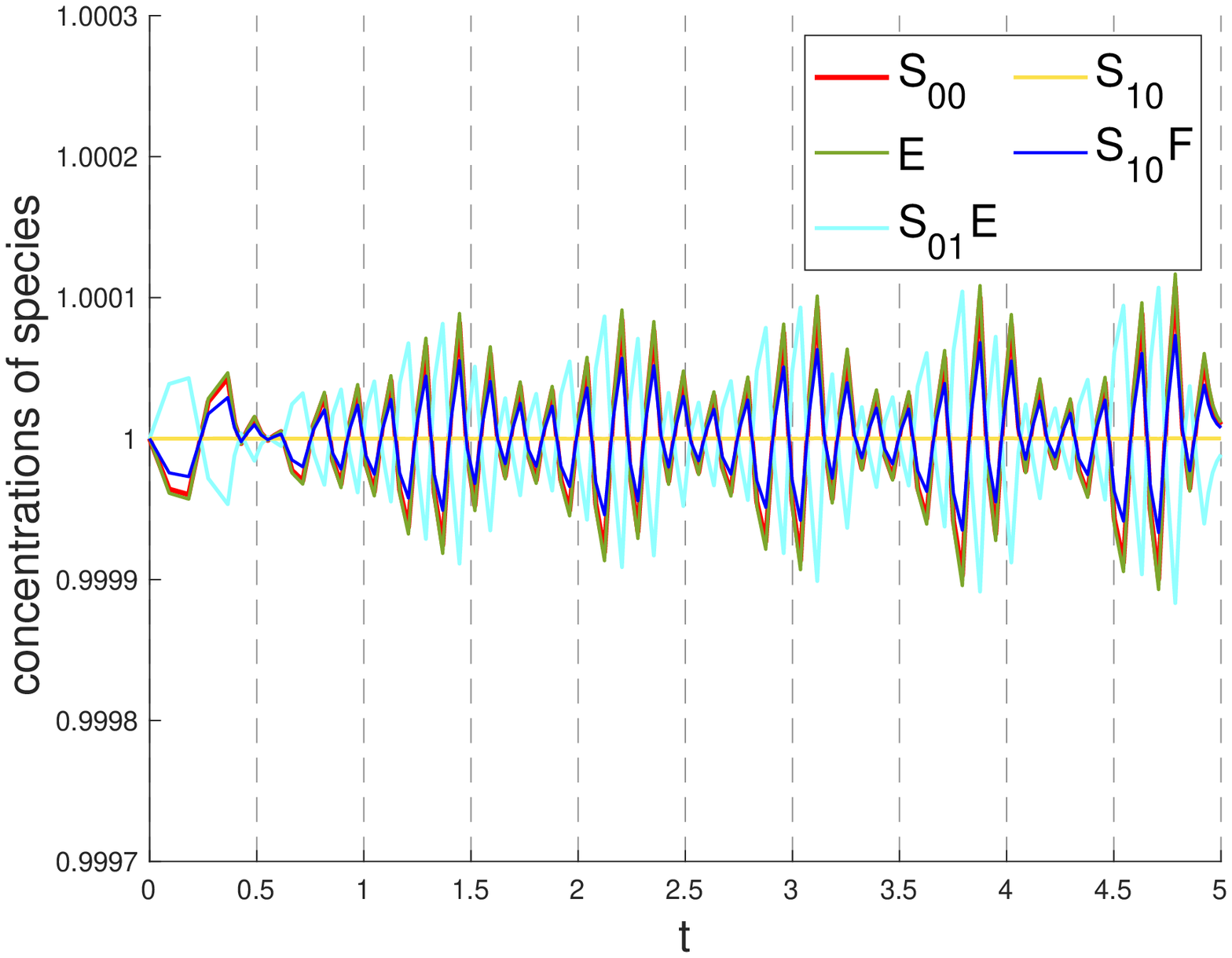}}
\hspace{0.01\linewidth}
\subfigure{
\includegraphics[width=0.45\linewidth]{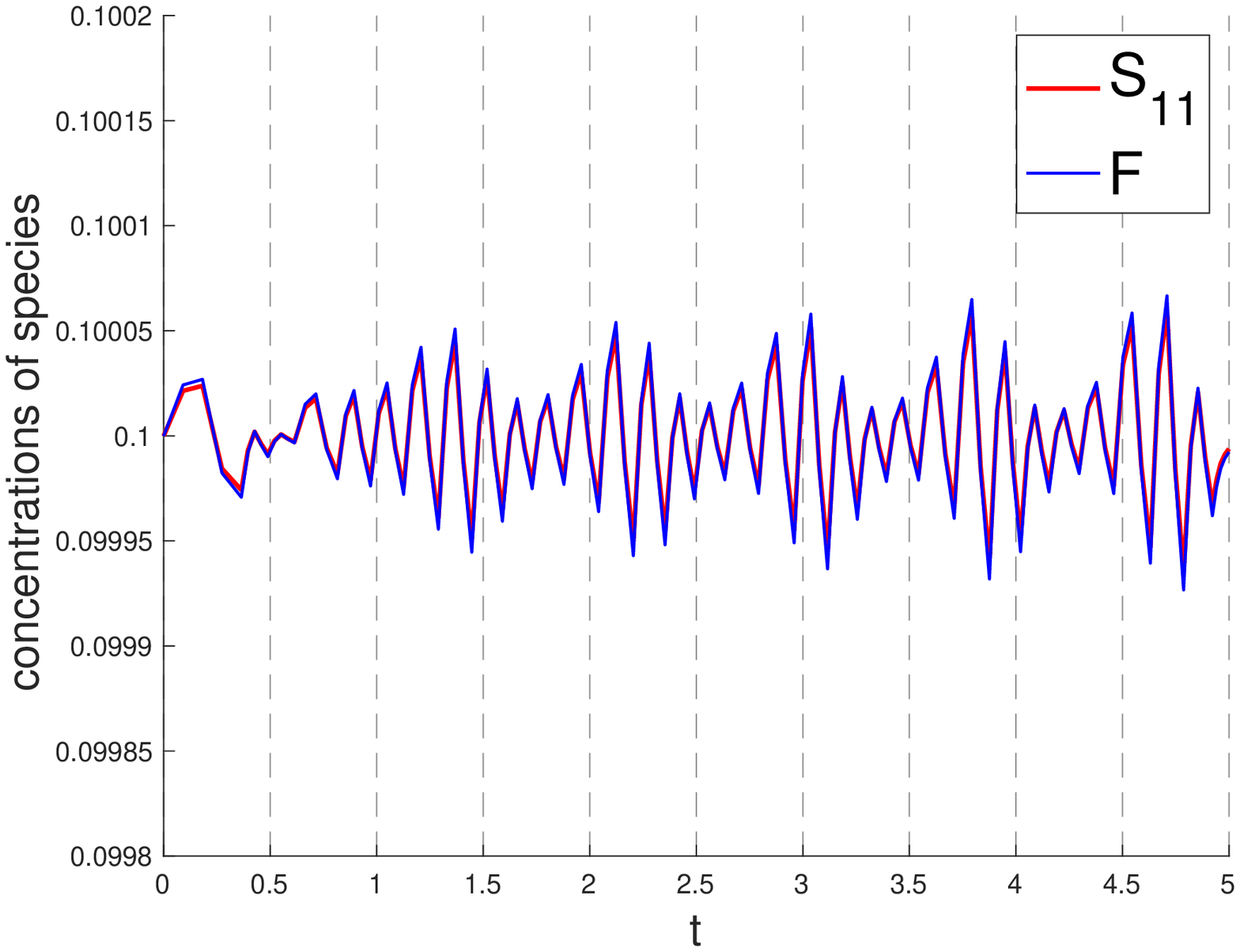}}
\caption{The network $\text{SubNet}_{\text{redERK}}$ gives rise to oscillations with respect to the species in Table \ref{tab:subredERK}. 
The left figure shows the changes of the concentrations for $S_{00},E,S_{01}E,S_{10}$ and $S_{10}F$.
The right figure shows the changes of the concentrations for $S_{11}$ and $F$.
The rate constants are given by $\tau(\tau_1)$ in \eqref{eq:tt1} where the value of $\tau_{1}$ is 3.08.
We set the initial concentrations to $y^{(0)}$ in \eqref{eq:y0}.}
\label{fig:subredERK}
\end{figure}

\noindent We observe that when the value of $\tau_1$ changes from 3.06 to 3.08, the real parts of a pair of conjugate-complex eigenvalues change from negative to positive, and become zero when $\tau_1=\tau^*_1$.
Moreover, all other eigenvalues remain with negative real parts.
So, by Definition \ref{def:kxHB}, $\text{SubNet}_{\text{redERK}}$ has a simple Hopf bifurcation at $(\tau^*,y^*)$ with respect to $\tau_1$.

Below, we show that  $\text{Net}_{\text{redERK}}$ also admits a simple Hopf bifurcation. 
Let $x_1,\ldots,x_{10}$ denote the concentrations of the species in $\text{Net}_{\text{redERK}}$, see Table \ref{tab:redERK}.
\begin{table}[ht]
\centering
\caption{Species concentrations of  $\text{Net}_{\text{redERK}}$.}
\begin{tabular}{cccccccccc}
\hline$x_{1}$ & $x_{2}$ & $x_{3}$ & $x_{4}$ & $x_{5}$ & $x_{6}$ & $x_{7}$ & $x_{8}$ & $x_{9}$ & $x_{10}$ \\
\hline$S_{00}$ & $E$ & $S_{01}E$ & $S_{11}$ & $S_{10}$ & $F$ & $S_{10}F$ & $S_{00}E$ & $S_{01}$ & $S_{11}F$ \\
\hline 
\end{tabular}
\label{tab:redERK}
\end{table}
Let $\kappa=(\kappa_1,\ldots,\kappa_{10})$, and let $x=(x_1,\ldots,x_{10})$.
Denote the system of ODEs  by $\dot{x}=f(\kappa,x)$.
Let ${\rm Jac}^{\rm red}_{f}(\kappa,x)$ be a reduced Jacobian matrix of $f(\kappa,x)$ with respect to $x$.
Let
\begin{align*}
\kappa^*=(1.0776,\;0.10776,\;0.0776,\;0.0776,\;110,\;0.11,\;0.1,\;0.1,\;1,\;1).
\end{align*}
Pick a positive steady state corresponding to $\kappa^*$:
\begin{align*}
x^*=(1,\;1,\;1,\;0.1,\;1,\;0.1,\;1,\;10,\;10,\;10).
\end{align*}
Again, the point $(\kappa^*,x^*)$ is obtained by the condition (i) in Lemma \ref{lemma:Yang}.
The eigenvalues of ${\rm Jac}^{\rm red}_f(\kappa^*,x^*)$ are approximately
\begin{align*}
-21.997,\;-2.0640,\;-1.274,\;-1.127,\;-0.195,\;\textbf{0.0999}\mathrm{i},\;\textbf{0.0999}\mathrm{i}.
\end{align*}
As the value of $\kappa_{3}$ varies, in the neighborhood of $\kappa_3^*$, consider the curve of  rate constants
\begin{align}\label{eq:kk3subredERK}
\kappa(\kappa_{3}) = (1+\kappa_{3},\;0.1+0.1\kappa_{3},\;\kappa_{3},\;\kappa_{3},\;110,\;0.11,\;0.1,\;0.1,\;1,\;1).
\end{align}
Notice again that all points on the above curve admit a common positive steady state  $x^*$.
When $\kappa_{3}=0.05$, the eigenvalues of ${\rm Jac}^{\rm red}_f(\kappa(\kappa_{3}),x^*)$ are approximately
{\small \begin{align*}
-21.997,-2.0189,-1.249,-1.108,-0.176,\textbf{-0.0103+0.0841}\mathrm{i},\textbf{-0.0103-0.0841}\mathrm{i}.
\end{align*}
}When $\kappa_{3}=0.1$, the eigenvalues of ${\rm Jac}^{\rm red}_f(\kappa(\kappa_{3}),x^*)$ are approximately
{\small \begin{align*}
-21.997,-2.101,-1.300,-0.138,-0.208,\textbf{0.00654+0.110}\mathrm{i},\textbf{0.00654-0.110}\mathrm{i}.
\end{align*}
}In this case, a nearby oscillation is generated around the positive steady state $x^*$, see figure \ref{fig:redERK}, where we set the initial concentrations to 
\begin{align}\label{eq:x0redERK}
x^{(0)}=(1,\;1,\;1,\;0.0999999,\;1,\;0.0999999,\;1,\;10,\;10,\;10).
\end{align} 
Notice that here,  we see that the capacity of $\text{SubNet}_{\text{redERK}}$ for oscillations is inherited from  $\text{Net}_{\text{redERK}}$ by the facts that $y^*$ exactly gives the first seven coordinates of $x^*$, and $y^{(0)}$ exactly gives the first seven coordinates of $x^{(0)}$. 
\begin{figure}[htbp]
\centering
\subfigure{
\includegraphics[width=0.45\linewidth]{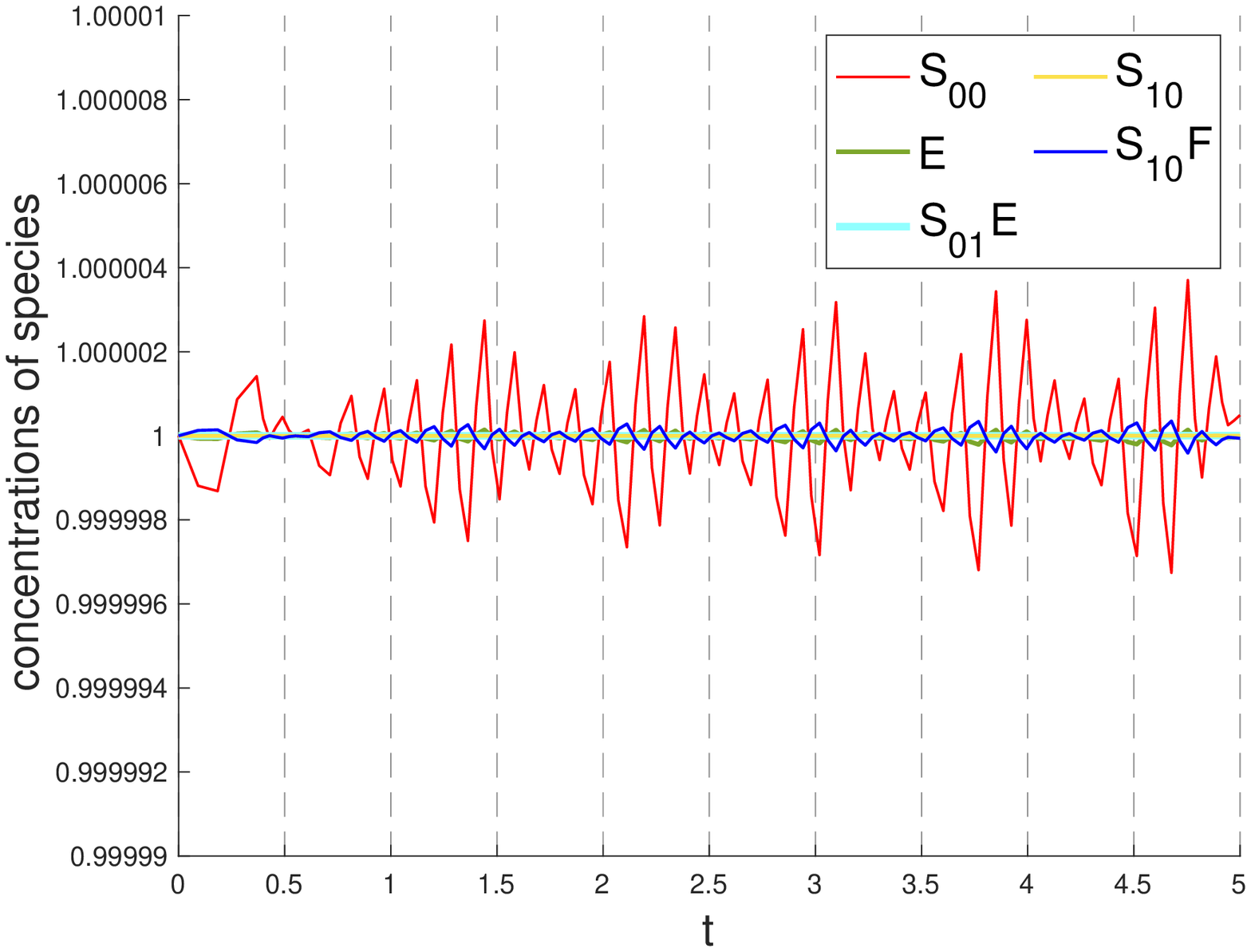}}
\hspace{0.01\linewidth}
\subfigure{
\includegraphics[width=0.45\linewidth]{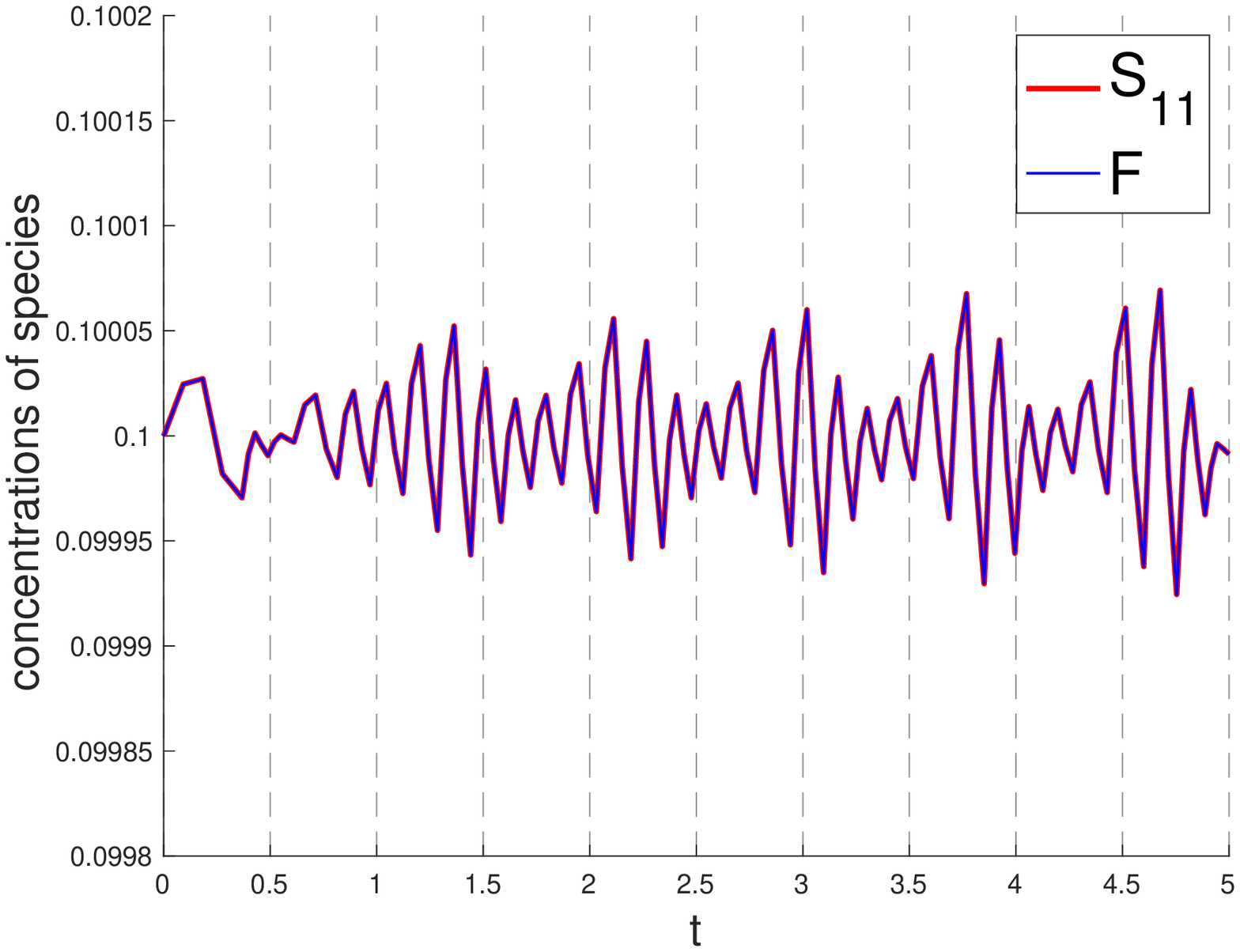}}
\caption{The network $\text{Net}_{\text{redERK}}$ gives rise to oscillations with respect to the species in Table \ref{tab:redERK}. 
The left figure shows the changes of the concentrations for $S_{00},E,S_{01}E,S_{10}$ and $S_{10}F$.
The right figure shows the changes of the concentrations for $S_{11}$ and $F$.
The rate constants are given by $\kappa(\kappa_{3})$ in \eqref{eq:kk3subredERK} where the value of $\kappa_{3}$ is 0.1.
We set the initial concentrations to $x^{(0)}$ in \eqref{eq:x0redERK}.}
\label{fig:redERK}
\end{figure}

\noindent We observe that when the value of $\kappa_{4}$ changes from 0.05 to 0.1, the real parts of a pair of conjugate-complex eigenvalues change from negative to positive, and become zero when $\kappa_{3}=\kappa_{3}^*$.
Moreover, all other eigenvalues remain with negative real parts.
So, by Definition \ref{def:kxHB}, $\text{Net}_{\text{redERK}}$ has a simple Hopf bifurcation at $(\kappa^*,x^*)$ with respect to $\kappa_{4}$.

\section{Transformation of the Jacobian matrix}\label{sec:trans}
According to Definition \ref{def:kxHB},
in order to determine if a network admits a Hopf bifurcation, we need to analyze 
the eigenvalues of a reduced Jacobian matrix 
 for all possible values of rate constants and their corresponding
 nondegenerate 
 positive steady states.
 Notice that 
 a reduced Jacobian matrix and the Jacobian matrix  share the same nonzero eigenvalues. 
In this section, 
we first reparameterize the Jacobian matrix using extreme rays (here, the idea is motivated by the approach  introduced in \cite{Conradi2019}), 
and then we give a criterion for precluding Hopf bifurcations by the transformed Jacobian matrices (see Lemma \ref{lemma:preHB}).
We also show how to use the transformed Jacobian matrix to detect Hopf bifurcations by an example (see Example \ref{ex:details}). 

Consider a network $G$ \eqref{eq:network} with a stoichiometric matrix $\mathcal{N}$ and a reactant matrix $\mathcal{Y}$. 
Recall the system $\dot{x}=f(\kappa,x)=\mathcal{N} v(\kappa, x)$ defined in \eqref{sys:f}. 
Note that for any $k \in \{1,\ldots,s\}$, and for any $\ell \in \{1,\ldots,m\}$, the $(k,\ell)$-entry  of the Jacobian matrix of $v(\kappa, x)$ with respect to $x$ is
\begin{align*}
\frac{\partial v_{k}(\kappa, x)}{\partial x_{\ell}}=v_{k}\left(\kappa, x\right) \alpha_{\ell k} \frac{1}{x_{\ell}}.
\end{align*}
For any vector $y:=(y_1, \ldots, y_n) \in \R^n$, we denote by $\diag(y)$ the $n \times n$ diagonal matrix with $y_i$ on the diagonal.
Then, the Jacobian matrix of $f(\kappa,x)$ with respect to $x$ (denoted by ${\rm Jac}_f(\kappa,x)$)  can be written as 
\begin{align}\label{eq:jacfkx}
{\rm Jac}_f(\kappa,x)=\mathcal{N}\diag  (v(\kappa, x))\mathcal{Y}^\top \diag (\frac{1}{x}),
\end{align}
where $\frac{1}{x}:=(\frac{1}{x_1},\ldots,\frac{1}{x_s})^\top$.



Next, we consider a transformation of the Jacobian matrix ${\rm Jac}_f(\kappa,x)$ evaluated at the positive steady states. 
Let $R^{(1)},\ldots,R^{(t)} \in \R^m_{\geq 0}$ be the extreme rays of the flux cone $F(\mathcal{N})$.
We introduce the new variables $h_1,\ldots,h_s,\lambda_1,\ldots,\lambda_t$. 
Let $h=(h_1,\ldots,h_s)$, and let $\lambda=(\lambda_1,\ldots,\lambda_t)$.
We define a new matrix in terms of $h$ and $\lambda$:
\begin{align}\label{eq:defjhl}
J(h,\lambda)&:= \mathcal{N}\diag  (\sum_{i=1}^{t} \lambda_{i} R^{(i)})\mathcal{Y}^\top \diag (h).
\end{align}


\begin{lemma}\label{prop:kxtohl}
Let $G$ be a network (see \eqref{eq:network}) with a stoichiometric matrix $\mathcal{N}$. 
Denote the system of ODEs by $\dot{x}=f(\kappa,x)=\mathcal{N}v(\kappa,x)$ as in \eqref{sys:f}.
Let ${\rm Jac}_f(\kappa,x)$ be the Jacobian matrix of $f(\kappa,x)$ with respect to $x$.
Let $J(h,\lambda) \in \Q[h,\lambda]^{s \times s}$ be the matrix corresponding to $G$ defined in \eqref{eq:defjhl}. 
For any $\kappa^* \in \R^m_{>0}$ and corresponding positive steady state $x^* \in \R^s_{>0}$, there exist $h^* \in \mathbb{R}_{> 0}^s$ and $\lambda^* \in \R^t_{\geq 0}$ such that $J(h^*,\lambda^*)={\rm Jac}_f(\kappa^*,x^*)$.
\end{lemma}
\begin{proof}
For any $\kappa^* \in \R^m_{>0}$ and for any  corresponding positive steady state $x^* \in \R^s_{>0}$, we have the equality
\begin{align*}
\mathcal{N} v(\kappa^*, x^*) = {\bf 0}.
\end{align*}
So, we have $v(\kappa^*, x^*) \in F(\mathcal{N})$.  
Let $R^{(1)},\ldots,R^{(t)} \in \R^m_{\geq 0}$ be the extreme rays of $F(\mathcal{N})$.
Then, there exist $\lambda_{1}^*,\ldots,\lambda_{t}^* \in {\mathbb R}_{\geq 0}$ such that
\begin{align}\label{eq:v}
v(\kappa^*, x^*)=\sum_{i=1}^{t} \lambda_{i}^* R^{(i)}.
\end{align}
Let $h^*=(\frac{1}{x^*_1},\ldots,\frac{1}{x^*_s})$, and let $\lambda^*=(\lambda^*_1,\ldots,\lambda^*_t)$.
By \eqref{eq:jacfkx}, \eqref{eq:defjhl}, and \eqref{eq:v}, we have
{\footnotesize
$J(h^*,\lambda^*) = \mathcal{N}\diag  (\sum_{i=1}^{t} \lambda_{i}^* R^{(i)})\mathcal{Y}^\top \diag (h^*)=\mathcal{N}\diag  (v(\kappa^*, x^*))\mathcal{Y}^\top \diag (\frac{1}{x^*})={\rm Jac}_f(\kappa^*,x^*)$.
}
\end{proof}

\begin{lemma}\label{lemma:preHB}
Let $G$ be a network (see \eqref{eq:network}) with a stoichiometric matrix $\mathcal{N}$. 
Denote the system of ODEs by $\dot{x}=f(\kappa,x)$ as in \eqref{sys:f}. 
Let $J(h,\lambda) \in \Q[h,\lambda]^{s \times s}$ be the matrix corresponding to $G$ defined in \eqref{eq:defjhl}. 
Let $r={\rm rank}({\mathcal N})$. 
We define
\begin{align}\label{def:qhlz}
q(h,\lambda;z):=\frac{1}{z^{s-r}}\operatorname{det}(z I-J(h,\lambda))=z^{r}+b_{1}(h,\lambda) z^{r-1}+\cdots+b_{r}(h,\lambda).
\end{align}
Let $H_{r-1}(h,\lambda)$ be the $(r-1)$-th Hurwitz matrix of $q(h,\lambda;z)$.
If for any $h^* \in \mathbb{R}_{> 0}^s$ and for any $\lambda^* \in \mathbb{R}_{\geq 0}^t$, $\operatorname{det}(H_{r-1}(h^*,\lambda^*))>0$ whenever $b_{r}(h^*,\lambda^*) \neq 0$ , then the network $G$ does not admit a Hopf bifurcation.
\end{lemma}
\begin{proof}
By Lemma \ref{prop:kxtohl}, for any $\kappa^* \in \R^m_{>0}$, and for any positive steady state $x^* \in \R^s_{>0}$ corresponding to $\kappa^*$, there exist $h^* \in \mathbb{R}_{> 0}^s$ and $\lambda^* \in \R^t_{\geq 0}$ such that $J(h^*,\lambda^*)={\rm Jac}_f(\kappa^*,x^*)$.
Therefore, the characteristic polynomial of $J(h^*,\lambda^*)$ is equal to the characteristic polynomial of ${\rm Jac}_f(\kappa^*,x^*)$.
Let $\phi(\kappa,x;z)$ be the characteristic polynomial of ${\rm Jac}^{\text{red}}_f(\kappa,x)$ (see \eqref{def:reduceJac}) and we write
\begin{align*}
\phi(\kappa,x;z) = z^{r}+a_{1}(\kappa,x) z^{r-1}+\cdots+a_{r}(\kappa,x).
\end{align*}
Note that $z^{s-r}\phi(\kappa,x;z)$ is the characteristic polynomial of ${\rm Jac}_f(\kappa,x)$.
So, we have
\begin{align}\label{eq:bi=ai}
b_{i}(h^*,\lambda^*) = a_{i}(\kappa^*,x^*),\text{ for any }1 \le i \le r.
\end{align}
Let $W_{r-1}(\kappa,x)$ be the $(r-1)$-th Hurwitz matrix of $\phi(\kappa,x;z)$.
By Definition \ref{def:hurwitz} and by \eqref{eq:bi=ai}, we have 
\begin{align*}
\operatorname{det}(H_{r-1}(h^*,\lambda^*))=\operatorname{det}(W_{r-1}(\kappa^*,x^*)).
\end{align*}
Then, by the hypothesis that $\operatorname{det}(H_{r-1}(h^*,\lambda^*))>0$ whenever $b_{r}(h^*,\lambda^*) \neq 0$, we have $\operatorname{det}(W_{r-1}(\kappa^*,x^*))>0$ whenever $a_{r}(\kappa^*,x^*) \neq 0$.
By Lemma \ref{prop:polyimroots}, there are no pure imaginary pairs of eigenvalues in ${\rm Jac}^{\text{red}}_f(\kappa^*,x^*)$.
Then, by Definition \ref{def:kxHB}, the network $G$ does not admit a Hopf bifurcation.
\end{proof}

\begin{lemma}\label{lemma:hltokx}
Let $G$ be a network (see \eqref{eq:network}) with a stoichiometric matrix $\mathcal{N}$.
Denote the system of ODEs by $\dot{x}=f(\kappa,x)$ as in \eqref{sys:f}.
Let ${\rm Jac}_f(\kappa,x)$ be the Jacobian matrix of $f(\kappa,x)$ with respect to $x$.
Let $R^{(1)},\ldots,R^{(t)} \in \R^m_{\geq 0}$ be the extreme rays of $F(\mathcal{N})$.
Let $J(h,\lambda) \in \Q[h,\lambda]^{s \times s}$ be the matrix corresponding to $G$ defined in \eqref{eq:defjhl}. 
For any $h^* \in \mathbb{R}_{> 0}^s$, and for any $\lambda^* \in \R^t_{\geq 0}$ such that $\sum_{i=1}^{t} \lambda_{i}^* R^{(i)} \in \R^m_{>0}$, there exist $\kappa^* \in \R^m_{>0}$ and a corresponding positive steady state $x^* \in \R^s_{>0}$ such that ${\rm Jac}_f(\kappa^*,x^*)=J(h^*,\lambda^*)$.
\end{lemma}
\begin{proof}
Let $x^* = (\frac{1}{h_1^*},\ldots, \frac{1}{h_s^*})$.
Denote by $\psi(i)$ the $i$-th column of ${\mathcal Y}$.
For any two vectors $a,b \in \mathbb R^{n}$, we denote by $a^{b}$ the product $\prod_{i=1}^{n} a_{i}^{b_{i}}$. 
Let 
\begin{align}\label{eq:kappa*}
\kappa^* = \diag \left( (h^*)^{\psi(1)},\ldots,(h^*)^{\psi(m)} \right) \sum_{i=1}^{t} \lambda_{i}^* R^{(i)}.
\end{align}
By the fact that $h^* \in \mathbb{R}_{> 0}^s$, for any $i \in \{1,\ldots,m\}$, we have $(h^*)^{\psi(i)}>0$.
Recall that $\sum_{i=1}^{t} \lambda_{i}^* R^{(i)} \in \R^m_{>0}$.
So, by \eqref{eq:kappa*}, we have $\kappa^* \in \R^m_{>0}$.
By \eqref{eq:v1...vm}, \eqref{eq:vj}, and \eqref{eq:kappa*}, we have
\begin{align}\label{eq:vk*x*}
v(\kappa^*, x^*) &= (\kappa_1^*(x^*)^{\psi(1)},\ldots,\kappa_m^*(x^*)^{\psi(m)})^\top \notag \\
&= \diag \left( (x^*)^{\psi(1)},\ldots,(x^*)^{\psi(m)} \right) \kappa^* \notag \\
&= \diag \left( (x^*)^{\psi(1)},\ldots,(x^*)^{\psi(m)} \right) \diag \left( (h^*)^{\psi(1)},\ldots,(h^*)^{\psi(m)} \right) \sum_{i=1}^{t} \lambda_{i}^* R^{(i)} \notag \\
&= \diag \left( (x^*h^*)^{\psi(1)},\ldots,(x^*h^*)^{\psi(m)} \right) \sum_{i=1}^{t} \lambda_{i}^* R^{(i)} \notag \\
&= \sum_{i=1}^{t} \lambda_{i}^* R^{(i)}.
\end{align}
Recall that for any $i \in \{1,\ldots,t\}$, $\mathcal{N} R^{(i)} = {\bf 0}$.
By \eqref{eq:vk*x*}, we have
\begin{align*}
\mathcal{N} v(\kappa^*, x^*) = \mathcal{N} \sum_{i=1}^{t} \lambda_{i}^* R^{(i)} = \sum_{i=1}^{t} \lambda_{i}^* \mathcal{N} R^{(i)}={\bf 0}.
\end{align*}
So, $x^*$ is a positive steady state with respect to $\kappa^*$.
By \eqref{eq:jacfkx}, \eqref{eq:defjhl}, and \eqref{eq:vk*x*}, we have\\
{\footnotesize
${\rm Jac}_f(\kappa^*,x^*) = \mathcal{N}\diag  (v(\kappa^*, x^*))\mathcal{Y}^\top \diag (\frac{1}{x^*})
= \mathcal{N}\diag  (\sum_{i=1}^{t} \lambda_{i}^* R^{(i)})\mathcal{Y}^\top \diag (h^*)
= J(h^*,\lambda^*).$
}
\end{proof}

\begin{example}[Network \eqref{network:example} continued]\label{ex:details}
We show how to find the 
parameters such that  
the network \eqref{network:example} has a Hopf bifurcation by 
  Lemma \ref{lemma:Yang} and Lemma \ref{lemma:hltokx}.

First, it is straightforward to write down the stoichiometric matrix $\mathcal{N}$ and the reactant matrix $\mathcal{Y}$ of the network \eqref{network:example}:
\begin{align*}
\mathcal{N}=
\left( \begin{matrix}
-1 & 1 & 0 & 0 & 0 \\
0 & 0 & 1 & -1 & -1 \\
0 & 0 & 0 & 1 & -1 \\
0 & 1 & 0 & -1 & 0
\end{matrix} \right), \text{ and  }
\mathcal{Y}=
\left( \begin{matrix}
1 & 0 & 1 & 0 & 1 \\
1 & 0 & 0 & 1 & 1 \\
1 & 1 & 0 & 0 & 1 \\
0 & 0 & 0 & 1 & 1
\end{matrix} \right).
\end{align*}
The stoichiometric matrix $\mathcal{N}$ yields the following extreme ray:
\begin{align*}
R=(1,1,2,1,1)
\end{align*}
(note that there exists only one extreme ray up to the natural equivalence). 
We introduce new variables $h_1,\ldots,h_4,\lambda$. 
Let $h=(h_1,\ldots,h_4)$.
By \eqref{eq:jacfkx},  the transformed Jacobian matrix $J(h,\lambda)$ is 
\begin{align*}
J(h,\lambda)=
\left( \begin{matrix}
-h_1 \lambda & -h_2 \lambda & 0 & 0 \\
h_1 \lambda & -2h_2 \lambda & -h_3 \lambda & -2h_4 \lambda \\
-h_1 \lambda & 0 & -h_3 \lambda & 0 \\
0 & -h_2 \lambda & h_3 \lambda & -h_4 \lambda
\end{matrix} \right).
\end{align*}
Denote the system of ODEs by $\dot{x}=f(\kappa,x)$. 
Let ${\rm Jac}_f(\kappa,x)$ be the Jacobian matrix of $f(\kappa,x)$ with respect to $x$.
For any $h^* \in \R^4_{>0}$ and for any $\lambda^* \in \R_{\geq 0}$ such that $\lambda^* R \in \R^m_{>0}$ (notice that in this example, any positive number $\lambda^*$ satisfies $\lambda^* R \in \R^m_{>0}$), by Lemma \ref{lemma:hltokx}, there exist $\kappa^* \in \R^5_{>0}$ and a corresponding positive steady state $x^* \in \R^4_{>0}$ such that
\begin{align}\label{eq:jacfkx=jhl}
{\rm Jac}_f(\kappa^*,x^*)=J(h^*,\lambda^*).
\end{align}
Notice that ${\rm rank} (\mathcal{N})=4$.
Suppose
\begin{align*}
\operatorname{det}(z I-J(h,\lambda))=zq(h,\lambda;z),
\end{align*}
where
\begin{align}\label{eq:qhlzinexp}
q(h,\lambda;z) =  z^{4}+b_{1}(h,\lambda) z^{3}+b_{2}(h,\lambda) z^{2}+b_{3}(h,\lambda) z+b_{4}(h,\lambda).
\end{align}
Suppose 
\begin{align*}
\operatorname{det}(z I-
{\rm Jac}_f(\kappa,x))=z\phi(\kappa,x;z),
\end{align*}
where
\begin{align}\label{eq:phikxzinexp}
\phi(\kappa,x;z)=z^{4}+a_{1}(\kappa,x) z^{3}+a_{2}(\kappa,x) z^{2}+a_{3}(\kappa,x) z+a_{4}(\kappa,x).
\end{align}
Recall that 
$\phi(\kappa, x; z)$ is the characteristic polynomial of 
a reduced Jacobian matrix (see Remark \ref{rmk:red}). 
By \eqref{eq:jacfkx=jhl}, we have $\phi(\kappa^*,x^*;z)=q(h^*,\lambda^*;z)$.
Then, by \eqref{eq:qhlzinexp} and \eqref{eq:phikxzinexp}, we have
\begin{align}\label{eq:a=b}
a_{i}(\kappa^*,x^*) = b_{i}(h^*,\lambda^*),\text{ for any }1 \le i \le 4.
\end{align}
For any positive integer $i\;(i\leq 4)$, let $W_{i}(\kappa,x)$ be the $i$-th Hurwitz matrix of $\phi(\kappa,x;z)$, and let $H_{i}(h,\lambda)$ be the $i$-th Hurwitz matrix of $q(h,\lambda;z)$.
By Definition \ref{def:hurwitz} and by \eqref{eq:a=b}, we have
\begin{align}\label{eq:detH=detW}
\operatorname{det}(H_{i}(h^*,\lambda^*))=\operatorname{det}(W_{i}(\kappa^*,x^*)),\text{ for any }1 \le i \le 4.
\end{align}
By \eqref{eq:a=b} and \eqref{eq:detH=detW}, in order to find $\kappa^* \in \R^5_{>0}$ and a corresponding positive steady state $x^* \in \R^4_{>0}$ that fulfill Lemma \ref{lemma:Yang} (i) and (ii), we only need to find $h^* \in \R^4_{>0}$ and $\lambda^* \in \R_{>0}$ that fulfill the following statements.
\begin{enumerate}[(I)]
\item $\operatorname{det}(H_{3}(h^*,\lambda^*))=0$ and $b_4(h^*,\lambda^*)>0$.\label{cond:I}
\item $\operatorname{det}(H_{1}(h^*,\lambda^*))>0 \text{ and } \operatorname{det}(H_{2}(h^*,\lambda^*))>0$.\label{cond:II}
\end{enumerate}
It is straightforward to compute that the polynomials ${\rm det}(H_{1}(h,\lambda))$, ${\rm det}(H_{2}(h,\lambda))$, and $b_4(h,\lambda)$ contain only terms with positive coefficients.
So, ${\rm det}(H_{1}(h^*,\lambda^*))>0$, ${\rm det}(H_{2}(h^*,\lambda^*))>0$, and $b_4(h^*,\lambda^*)>0$.
In other words, we only need to find $h^*$ and $\lambda^*$ such that ${\rm det}(H_{3}(h^*,\lambda^*))=0$.
We can compute that  
{\small
\begin{align}\label{eq:detH3expand}
\operatorname{det}(H_{3}(h,\lambda)) =& 4h_1 h_2 h_3 \lambda^6 (3h_1^2 h_2+6h_1 h_2^2+h_1^2 h_3+3h_1 h_2 h_3+4h_2^2 h_3+h_1 h_3^2+2h_2 h_3^2)\\
& +h_1 h_2 h_4 \lambda^6 (3h_1^2 h_2+6h_1 h_2^2+h_1^2h_4+4h_1h_2h_4+h_1h_4^2) \notag \\
&+ h_1 h_3 h_4 \lambda^6 (h_1^2h_3+h_1h_3^2+h_1^2h_4+2h_1h_2h_3+h_3^2h_4+h_1h_4^2+h_3h_4^2) \notag  \\
& + h_1 h_2 h_3 h_4 \lambda^6 (4h_1+9h_1h_2+4h_1h_3+6h_2h_3+2h_3^2+2h_1h_4+h3_h4)  \notag \\
& -3h_1 h_2 h_3 h_4 \lambda^6(4h_2^2+4h_2h_4+h_4^2) \notag .
\end{align}
}We set $h_2=1,h_3=\frac{1}{2},h_4=1,\text{and }\lambda=1$.
The right-hand side of \eqref{eq:detH3expand} becomes
\begin{align*}
\frac{55}{4}h_1^3+\frac{273}{8}h_1^2-\frac{49}{8}h_1.
\end{align*}
We solve $\frac{55}{4}h_1^3+\frac{273}{8}h_1^2-\frac{49}{8}h_1=0$ by the function \texttt{Solve} in Mathematica.
It turns out that the only one positive solution is $h_1=\frac{7(\sqrt{1961}-39)}{220}$.
Let $h^*=(\frac{7(\sqrt{1961}-39)}{220},1,\frac{1}{2},1)$, and let $\lambda^*=1$.
Based on the discussion above, we know that  $h^*$ and $\lambda^*$ fulfill \eqref{cond:I} and \eqref{cond:II}.
By the proof of Lemma \ref{lemma:hltokx}, 
we can find the corresponding $\kappa^*$ and $x^*$:
\begin{align*}
\kappa^* = (\frac{7(\sqrt{1961}-39)}{440},\;\frac{1}{2},\;\frac{7(\sqrt{1961}-39)}{110},\;1,\;\frac{7(\sqrt{1961}-39)}{440}),
\end{align*}
and
\begin{align*}
x^*=(\frac{220}{7(\sqrt{1961}-39)},\;1,\;2,\;1).
\end{align*}
Note that it is easy to check that 
\begin{align*}
\left.\frac{d\left(\operatorname{det} (W_{3}(\kappa(\kappa_4),x(\kappa_4)))\right)}{d \kappa_4}\right|_{\kappa=\kappa^*,x=x^*} \approx -2.430 \ne 0,
\end{align*}
where $\kappa(\kappa_4)$ and $x(\kappa_4)$ are defined in 
\eqref{eq:kk4} and \eqref{eq:xk4}.
So, the point $(\kappa^*,x^*)$ also fulfills Lemma \ref{lemma:Yang} (iii).
Hence, the network \eqref{network:example} has a simple Hopf bifurcation at $(\kappa^*,x^*)$ with respect to $\kappa_4$.

\end{example}

\section{Structure of the Jacobian Matrix}\label{sec:structure}
Let $G$ be a zero-one network (see \eqref{eq:network}) with a reactant matrix $\mathcal{Y}$ and a stoichiometric matrix $\mathcal{N}$. 
Denote the system of ODEs by $\dot{x}=f(\kappa,x)$ (see \eqref{sys:f}). 
Let $R^{(1)},\ldots,R^{(t)} \in \R^m_{\geq 0}$ be the extreme rays of the flux cone $F(\mathcal{N})$, and let $J(h,\lambda)$ be the matrix defined in \eqref{eq:defjhl}.
For any $i \in \{1,\ldots,t\}$, we define a matrix in $\R^{s \times s}$
\begin{align}\label{eq:mui}
\mu^{(i)} := \mathcal{N}\diag  (R^{(i)}) \mathcal{Y}^\top.
\end{align}
Then, by \eqref{eq:defjhl},  we can rewrite $J(h,\lambda)$ as
{\footnotesize
\begin{align}\label{eq:jhlinm}
J(h,\lambda)= \mathcal{N}\diag  (\sum_{i=1}^{t} \lambda_{i} R^{(i)})\mathcal{Y}^\top \diag (h) = \sum_{i=1}^{t} \lambda_{i} \mathcal{N}\diag  (R^{(i)}) \mathcal{Y}^\top \diag (h) = \sum_{i=1}^{t} \lambda_{i} \mu^{(i)} \diag (h).
\end{align}
}

In the rest of this paper, for any matrix $M \in \R^{a \times b}$, for any $1\leq k \leq a$ and for any $1\leq \ell \leq b$, we denote by $M_{k \ell}$ the $(k,\ell)$-entry of $M$.

\begin{lemma}\label{lemma:coef}
Let $G$ be a zero-one network with a stoichiometric matrix $\mathcal{N}$. 
Let $R^{(1)},\ldots,R^{(t)} \in {\mathbb R}_{\geq 0}^m$ be the extreme rays of the flux cone $F(\mathcal{N})$. 
For any $i \in \{1,\ldots,t\}$, let $\mu^{(i)} \in \R^{s \times s}$ be the matrix defined in \eqref{eq:mui}. 
Then, for any $k \in \{1,\ldots,s\}$, and for any $\ell \in \{1,\ldots,s\}$, we have $\mu^{(i)}_{kk} \leq 0$ and $\lvert \mu^{(i)}_{k\ell} \rvert \leq -\mu^{(i)}_{kk}$.
\end{lemma}
\begin{proof}
For any $k \in \{1,\ldots,s\}$, we define
\begin{align}
\mathcal{N}_{+}(k)&:=\{j:\;\mathcal{N}_{kj}=1,\; 1 \leq j \leq m\} \text{, and}\label{eq:n+}\\ 
\mathcal{N}_{-}(k)&:=\{j:\;\mathcal{N}_{kj}=-1, \;1 \leq j \leq m\}\label{eq:n-}.
\end{align}
Here, recall that $\mathcal{N}_{kj}=\beta_{kj}-\alpha_{kj}$, and $\alpha_{kj},\beta_{kj} \in \{0,1\}$. 
So, we have the following statements.
\begin{enumerate}[(I)]
\item \label{I} If $j \in \mathcal{N}_{+}(k)$ (i.e., $\beta_{kj}-\alpha_{kj}=1$), then $\alpha_{kj}=0, \beta_{kj}=1$.
\item \label{II} If $j \in \mathcal{N}_{-}(k)$ (i.e., $\beta_{kj}-\alpha_{kj}=-1$), then $\alpha_{kj}=1, \beta_{kj}=0$.
\end{enumerate}
By \eqref{eq:mui}, \eqref{eq:n+}, \eqref{eq:n-}, \eqref{I}, and \eqref{II}, we have 
{\small\begin{align}\label{eq:muikk<=0}
\mu^{(i)}_{kk}&=\sum_{j=1}^m \mathcal{N}_{kj} \alpha_{kj} R^{(i)}_j =\sum_{j \in \mathcal{N}_{+}(k)} \mathcal{N}_{kj} \alpha_{kj} R^{(i)}_j + \sum_{j \in \mathcal{N}_{-}(k)} \mathcal{N}_{kj} \alpha_{kj} R^{(i)}_j  =-\sum_{j \in \mathcal{N}_{-}(k)} R^{(i)}_j \leq 0.
\end{align}}
For any $i \in \{1,\ldots,t\}$, by the fact that $R^{(i)} \in F(\mathcal{N})$, we have $\mathcal{N}R^{(i)}=\bf 0$.
Notice that the $k$-th coordinate of $\mathcal{N}R^{(i)}$ is $\sum_{j = 1}^m \mathcal{N}_{kj} R^{(i)}_j$.
Then, by \eqref{eq:n+} and \eqref{eq:n-}, we have
\begin{align*}
0=\sum_{j = 1}^m \mathcal{N}_{kj} R^{(i)}_j =\sum_{j \in \mathcal{N}_{+}(k)} \mathcal{N}_{kj} R^{(i)}_j + \sum_{j \in \mathcal{N}_{-}(k)} \mathcal{N}_{kj} R^{(i)}_j =\sum_{j \in \mathcal{N}_{+}(k)} R^{(i)}_j - \sum_{j \in \mathcal{N}_{-}(k)} R^{(i)}_j,
\end{align*} 
i.e., 
\begin{align}\label{eq:+=-}
\sum_{j \in \mathcal{N}_{+}(k)} R^{(i)}_j = \sum_{j \in \mathcal{N}_{-}(k)} R^{(i)}_j.
\end{align}
By \eqref{eq:mui}, \eqref{eq:n+}, \eqref{eq:n-}, \eqref{I}, \eqref{II}, \eqref{eq:muikk<=0}, and \eqref{eq:+=-}, for any $\ell \in \{1,\ldots,s\}$, we have
\begin{align}\label{eq:muikl}
\lvert \mu^{(i)}_{k\ell} \rvert&=\lvert \sum_{j=1}^m \mathcal{N}_{kj} \alpha_{\ell j} R^{(i)}_j \rvert \notag \\
&= \lvert \sum_{j \in \mathcal{N}_{+}(k)} \mathcal{N}_{kj} \alpha_{\ell j} R^{(i)}_j + \sum_{j \in \mathcal{N}_{-}(k)} \mathcal{N}_{kj} \alpha_{\ell j} R^{(i)}_j \rvert \notag \\
&= \lvert \sum_{j \in \mathcal{N}_{+}(k)} \alpha_{\ell j} R^{(i)}_j - \sum_{j \in \mathcal{N}_{-}(k)} \alpha_{\ell j} R^{(i)}_j \rvert \notag \\
&\leq \max \left(\sum_{j \in \mathcal{N}_{+}(k)} \alpha_{\ell j} R^{(i)}_j, \sum_{j \in \mathcal{N}_{-}(k)} \alpha_{\ell j} R^{(i)}_j \right) \notag \\
&\leq \sum_{j \in \mathcal{N}_{-}(k)} R^{(i)}_j = -\mu^{(i)}_{kk}.
\end{align}
Recall that $\alpha_{\ell j} \in \{0,1\}$ and $R^{(i)}_j \geq 0$.
The first inequality in \eqref{eq:muikl} holds because $\sum_{j \in \mathcal{N}_{+}(k)} \alpha_{\ell j} R^{(i)}_j \geq 0$ and $\sum_{j \in \mathcal{N}_{-}(k)} \alpha_{\ell j} R^{(i)}_j \geq 0$. 
The second inequality in \eqref{eq:muikl} holds because {\small $\sum_{j \in \mathcal{N}_{+}(k)} \alpha_{\ell j} R^{(i)}_j \leq \sum_{j \in \mathcal{N}_{+}(k)} R^{(i)}_j$, $\sum_{j \in \mathcal{N}_{-}(k)} \alpha_{\ell j} R^{(i)}_j \leq \sum_{j \in \mathcal{N}_{-}(k)} R^{(i)}_j$}, and \eqref{eq:+=-}. 
The last equality in \eqref{eq:muikl} holds because \eqref{eq:muikk<=0}.
\end{proof}

\begin{corollary}\label{coro:jkkis-}
Let $G$ be a zero-one network. 
Let $J(h,\lambda) \in \Q[h,\lambda]^{s \times s}$ be the matrix corresponding to $G$ defined in \eqref{eq:defjhl}. 
Then, for any $k \in \{1,\ldots,s\}$, $J(h,\lambda)_{kk}$ is either a zero polynomial or a sum of terms with negative coefficients.
\end{corollary}
\begin{proof}
By \eqref{eq:jhlinm}, for any $k \in \{1,\ldots,s\}$, we have
\begin{align*}
J(h,\lambda)_{kk} = \sum_{i=1}^{t} \mu^{(i)}_{kk} \lambda_{i} h_k.
\end{align*}
By Lemma \ref{lemma:coef}, for any $i \in \{1,\ldots,t\}$, we have $\mu^{(i)}_{kk} \leq 0$.
Notice that $\mu^{(i)}_{kk}$ is the coefficient of the term $\mu^{(i)}_{kk} \lambda_{i} h_k$.
So, we have the result.
\end{proof}

\begin{lemma}\label{lemma:poly}
Let $G$ be a zero-one network. 
Let $J(h,\lambda) \in \Q[h,\lambda]^{s \times s}$ be the matrix corresponding to $G$ defined in \eqref{eq:defjhl}. 
Given a positive integer $n\;(n \leq s)$, for any permutation map $\sigma:\{1,\ldots,n\} \rightarrow \{1,\ldots,n\}$, and for any $U=\{u_1,\ldots,u_n\} \subset \{1,\ldots,s\}$, we define the following polynomial in $\mathbb{Q}[h,\lambda]$:
\begin{align}\label{eq:p}
\mathcal{P}_{\sigma,U}(h,\lambda) := (-1)^n\left(\prod_{i=1}^{n} J(h,\lambda)_{u_i,u_i} - \prod_{i=1}^{n} J(h,\lambda)_{u_i,\sigma(u_i)} \right).
\end{align}
Then, $\mathcal{P}_{\sigma,U}(h,\lambda)$ is either a zero polynomial or a sum of terms with positive coefficients.


\end{lemma}
\begin{proof}
(i) By \eqref{eq:jhlinm}, for any $k \in \{1,\ldots,s\}$ and for any $l \in \{1,\ldots,s\}$, we have
\begin{align}\label{eq:jkl}
J(h,\lambda)_{kl}=\sum_{i=1}^{t} \mu^{(i)}_{kl} \lambda_{i} h_l.
\end{align}
We define
\begin{align}\label{eq:p1}
P_1(\lambda) = \prod_{i=1}^{n} \left(\sum_{j=1}^{t} -\mu^{(j)}_{u_i,u_i} \lambda_{j} \right),
\end{align}
and
\begin{align}\label{eq:p2}
P_2(\lambda) = \prod_{i=1}^{n} \left(\sum_{j=1}^{t} -\mu^{(j)}_{u_i,\sigma(u_i)} \lambda_{j} \right).
\end{align}
By \eqref{eq:jkl}, \eqref{eq:p1}, and \eqref{eq:p2}, we have
\begin{align}\label{eq:temp1}
(-1)^n\prod_{i=1}^{n} J(h,\lambda)_{u_i,u_i}=\prod_{i=1}^{n} \left(\sum_{j=1}^{t} -\mu^{(j)}_{u_i,u_i} \lambda_{j} h_{u_i} \right) = P_1(\lambda) \prod_{i=1}^{n} h_{u_i},
\end{align}
and
\begin{align}\label{eq:temp2}
(-1)^n\prod_{i=1}^{n} J(h,\lambda)_{u_i,\sigma(u_i)}=\prod_{i=1}^{n} \left(\sum_{j=1}^{t} -\mu^{(j)}_{u_i,\sigma(u_i)} \lambda_{j} h_{\sigma(u_i)} \right) = P_1(\lambda) \prod_{i=1}^{n} h_{\sigma(u_i)}.
\end{align}
Note that $\sigma$ is a permutation map.
Then, we have
\begin{align}\label{eq:temp3}
\prod_{i=1}^{n} h_{u_i}=\prod_{i=1}^{n} h_{\sigma(u_i)}.
\end{align}
We substitute \eqref{eq:temp1} and \eqref{eq:temp2} into \eqref{eq:p}.
Then, by \eqref{eq:temp3}, we have
\begin{align}\label{eq:phl}
\mathcal{P}_{\sigma,U}(h,\lambda) = \left(P_1(\lambda)-P_2(\lambda)\right) \prod_{i=1}^{n} h_{u_i}.
\end{align}
Next, we show that $P_1(\lambda)-P_2(\lambda)$ is either a zero polynomial or a sum of terms with positive coefficients.
We can rewrite the right-hand side of \eqref{eq:p1} as
\begin{align*}
(-\mu^{(1)}_{u_1,u_1}\lambda_1-\ldots-\mu^{(t)}_{u_1,u_1}\lambda_t)\times\cdots\times(-\mu^{(1)}_{u_n,u_n}\lambda_1-\ldots-\mu^{(t)}_{u_n,u_n}\lambda_t).
\end{align*}
Therefore, we can rewrite $P_1(\lambda)$ as
\begin{align}\label{eq:p1expand}
P_1(\lambda) = \sum_{(p_1,\ldots,p_n) \in \{1,\ldots,t\}^n} \prod_{i=1}^{n} \left(-\mu^{(p_i)}_{u_i,u_i}\lambda_{p_i}\right).
\end{align}
Similarly, we can rewrite $P_2(\lambda)$ as
\begin{align}\label{eq:p2expand}
P_2(\lambda) = \sum_{(p_1,\ldots,p_n) \in \{1,\ldots,t\}^n} \prod_{i=1}^{n} \left(-\mu^{(p_i)}_{u_i,\sigma(u_i)}\lambda_{p_i}\right).
\end{align}
We define
\begin{align}\label{eq:ap1pn}
a(p_1,\ldots,p_n):=\prod_{i=1}^{n} \left(-\mu^{(p_i)}_{u_i,u_i}\right)-\prod_{i=1}^{n} \left(-\mu^{(p_i)}_{u_i,\sigma(u_i)}\right).
\end{align}
By \eqref{eq:p1expand}, \eqref{eq:p2expand}, and \eqref{eq:ap1pn}, we have
\begin{align}\label{eq:p1-p2}
P_1(\lambda)-P_2(\lambda) =& \sum_{(p_1,\ldots,p_n) \in \{1,\ldots,t\}^n} \left( \prod_{i=1}^{n} \left(-\mu^{(p_i)}_{u_i,u_i}\lambda_{p_i}\right) - \prod_{i=1}^{n} \left(-\mu^{(p_i)}_{u_i,\sigma(u_i)}\lambda_{p_i}\right) \right) \notag \\
=& \sum_{(p_1,\ldots,p_n) \in \{1,\ldots,t\}^n} \left( \prod_{i=1}^{n} \left(-\mu^{(p_i)}_{u_i,u_i}\right)-\prod_{i=1}^{n} \left(-\mu^{(p_i)}_{u_i,\sigma(u_i)}\right)\right) \prod_{i=1}^{n}\lambda_{p_i} \notag \\
=& \sum_{(p_1,\ldots,p_n) \in \{1,\ldots,t\}^n} a(p_1,\ldots,p_n) \prod_{i=1}^{n}\lambda_{p_i}.
\end{align}
By Lemma \ref{lemma:coef}, for any $i \in \{1,\ldots,n\}$, $\lvert \mu^{(p_i)}_{u_i,\sigma(u_i)} \rvert \leq -\mu^{(p_i)}_{u_i,u_i}$, so we have
\begin{align*}
\prod_{i=1}^{n} \left(-\mu^{(p_i)}_{u_i,\sigma(u_i)}\right) \leq \prod_{i=1}^{n} \left(-\mu^{(p_i)}_{u_i,u_i}\right).
\end{align*}
It indicates that for any $(p_1,\ldots,p_n) \in \{1,\ldots,t\}^n$, $a(p_1,\ldots,p_n) \geq 0$.
Notice that in \eqref{eq:p1-p2}, $a(p_1,\ldots,p_n)$ is the coefficient of the term $ a(p_1,\ldots,p_n) \prod_{i=1}^{n}\lambda_{p_i}$.
Therefore, $P_1(\lambda)-P_2(\lambda)$ is either a zero polynomial or a sum of terms with positive coefficients.
Then, by \eqref{eq:phl}, $\mathcal{P}_{\sigma,U}(h,\lambda)$ is either a zero polynomial or a sum of terms with positive coefficients.
\end{proof}

Given a matrix $M \in \mathbb{R}^{s \times s}$ and a set $\{i_1,\ldots,i_k\} \subset \{1,\ldots,s\}$, we denote by $M[i_1,\ldots,i_k]$ the $k \times k$ matrix obtained by deleting the rows and columns of $M$ indexed by $\{1,\ldots,s\} \setminus \{i_1,\ldots,i_k\}$.

Recall the matrix $J(h,\lambda)$ defined in \eqref{eq:defjhl}. 
Here, we simply denote $J(h,\lambda)$ by $J$.
For any $1 \leq i<j<k \leq s$, we define
\begin{align}\label{def:dijk}
d_{i, j, k} := J_{ij}J_{jk}J_{ki} + J_{ik}J_{ji}J_{kj}-2J_{ii}J_{jj}J_{kk}.
\end{align}

\begin{lemma}\label{lemma:positivePM}
Let $G$ be a zero-one network. 
Let $J(h,\lambda) \in \Q[h,\lambda]^{s \times s}$ be the matrix corresponding to $G$ defined in \eqref{eq:defjhl}. 
Let $d_{i, j, k}$ be defined as in \eqref{def:dijk}.
Then, we have the following statements.
\begin{enumerate}[(i)]
\item For any $1 \leq i<j \leq s$, ${\rm det}(J[i,j])$ is either a zero polynomial or a sum of terms with positive coefficients.
\item For any $1 \leq i<j<k \leq s$, $d_{i, j, k}$ is either a zero polynomial or a sum of terms with positive coefficients.
\end{enumerate}
\end{lemma}
\begin{proof}
(i) Let $U=\{i,j\}$, and let $\sigma:U \rightarrow U$ be a map such that $\sigma(i)=j$ and $\sigma(j)=i$.
Then, the polynomial \eqref{eq:p} in Lemma \ref{lemma:poly} can be written as 
\begin{align*}
\mathcal{P}_{\sigma,U}(h,\lambda) = J_{ii}J_{jj} - J_{ij}J_{ji}.
\end{align*}
Notice that $\operatorname{det} \left( J[i, j] \right)=J_{ii}J_{jj}-J_{ij}J_{ji}$.
So, by Lemma \ref{lemma:poly}, ${\rm det}(J[i,j])$ is either a zero polynomial or a sum of terms with positive coefficients.

(ii) Let $U=\{i,j,k\}$, and let $\sigma_1:U \rightarrow U$ be a map such that $\sigma_1(i)=j$, $\sigma_1(j)=k$, and $\sigma_1(k)=i$.
Then, the polynomial \eqref{eq:p} in Lemma \ref{lemma:poly} can be written as
\begin{align*}
\mathcal{P}_{\sigma_1,U}(h,\lambda) = J_{ij}J_{jk}J_{ki}-J_{ii}J_{jj}J_{kk}.
\end{align*}
By Lemma \ref{lemma:poly}, $J_{ij}J_{jk}J_{ki}-J_{ii}J_{jj}J_{kk}$ is either a zero polynomial or a sum of terms with positive coefficients. 
Similarly, if we let $\sigma_2:U\rightarrow U$ be another map such that $\sigma_2(i)=k$, $\sigma_2(j)=i$, and $\sigma_2(k)=j$, then the polynomial \eqref{eq:p} in Lemma \ref{lemma:poly} can be written as
\begin{align*}
\mathcal{P}_{\sigma_2,U}(h,\lambda) = J_{ik}J_{ji}J_{kj}-J_{ii}J_{jj}J_{kk}.
\end{align*}
By Lemma \ref{lemma:poly}, $J_{ik}J_{ji}J_{kj}-J_{ii}J_{jj}J_{kk}$ is either a zero polynomial or a sum of terms with positive coefficients.
Notice that $d_{i, j, k}=\mathcal{P}_{\sigma_1,U}(h,\lambda)+\mathcal{P}_{\sigma_2,U}(h,\lambda)$.
So, $d_{i, j, k}$ is either a zero polynomial or a sum of terms with positive coefficients.
\end{proof}

\section{Proof of Theorem \ref{thm:rank<3}}\label{sec:proofmain}
\begin{lemma}\label{lemma:jpp<0jqq<0}
Let $G$ be a zero-one network. 
Let $J(h,\lambda) \in \Q[h,\lambda]^{s \times s}$ be the matrix corresponding to $G$ defined in \eqref{eq:defjhl}. 
For any $h^* \in \mathbb{R}_{> 0}^s$ and for any $\lambda^* \in \mathbb{R}_{\geq 0}^t$, if there exist $p$ and $q$ ($1 \leq p<q \leq s$) such that $\operatorname{det} \left( J(h^{*},\lambda^*)[p, q] \right) > 0$, then $J(h^*,\lambda^*)_{pp}<0$ and $J(h^*,\lambda^*)_{qq}<0$.
\end{lemma}
\begin{proof}
By \eqref{eq:jhlinm}, for any $k \in \{1,\ldots,s\}$ and for any $\ell \in \{1,\ldots,s\}$, we have
\begin{align}\label{eq:jhlkl}
J(h,\lambda)_{k\ell} = \sum_{i=1}^{t} \mu^{(i)}_{k\ell} \lambda_{i} h_\ell.
\end{align}
Then, we have
\begin{align}\label{eq:detjpqhl}
\operatorname{det} \left( J(h^{*},\lambda^*)[p, q] \right) =& \left | \begin{matrix}
J(h^{*},\lambda^*)_{pp} & J(h^{*},\lambda^*)_{pq}  \\
J(h^{*},\lambda^*)_{qp} & J(h^{*},\lambda^*)_{qq}  \\
\end{matrix} \right | \notag \\
=& 
\left | \begin{matrix}
\sum_{i=1}^{t} \lambda_{i}^* \mu_{pp}^{(i)} & \sum_{i=1}^{t} \lambda_{i}^* \mu_{pq}^{(i)}  \\
\sum_{i=1}^{t} \lambda_{i}^* \mu_{qp}^{(i)} & \sum_{i=1}^{t} \lambda_{i}^* \mu_{qq}^{(i)}  \\
\end{matrix} \right |
\cdot
\left | \begin{matrix}
h_p^* & 0  \\
0 & h_q^*  \\
\end{matrix} \right | \notag \\
=& h_p^* h_q^* \sum_{i=1}^{t}\sum_{j=1}^{t} \lambda_{i}^*\lambda_{j}^* \left( \mu_{pp}^{(i)}\mu_{qq}^{(j)} - \mu_{pq}^{(i)}\mu_{qp}^{(j)} \right).
\end{align}
Recall that $h^* \in \mathbb{R}_{> 0}^s$.
So, we have $h_p^*, h_q^* > 0$.
By \eqref{eq:detjpqhl} and by the hypothesis that $\operatorname{det} \left( J(h^{*},\lambda^*)[p, q] \right) > 0$, there exist $i^{\prime},j^{\prime}\in \{1,\ldots,t\}$ such that
\begin{align}\label{eq:llmm-mm}
\lambda^*_{i^{\prime}}\lambda^*_{j^{\prime}} \left( \mu_{pp}^{(i^{\prime})}\mu_{qq}^{(j^{\prime})} - \mu_{pq}^{(i^{\prime})}\mu_{qp}^{(j^{\prime})} \right) > 0.
\end{align}
Recall that $\lambda^* \in \mathbb{R}_{\geq 0}^s$.
So, by \eqref{eq:llmm-mm}, we have
\begin{align}\label{eq:l>0l>0}
\lambda^*_{i^{\prime}} > 0 \; \text{and} \; \lambda^*_{j^{\prime}} > 0,
\end{align}
and
\begin{align}\label{eq:mm-mm>0}
\mu_{pp}^{(i^{\prime})}\mu_{qq}^{(j^{\prime})} - \mu_{pq}^{(i^{\prime})}\mu_{qp}^{(j^{\prime})}>0.
\end{align}
By Lemma \ref{lemma:coef}, for any $i,j\in \{1,\ldots,t\}$, we have
\begin{align}\label{eq:m<=0m<=0}
\mu_{pp}^{(i)} \leq 0 \; \text{and} \; \mu_{qq}^{(i)} \leq 0,
\end{align}
and
\begin{align}\label{eq:muppqq>}
\mu_{pp}^{(i)}\mu_{qq}^{(j)} \geq \lvert \mu_{pq}^{(i)}\mu_{qp}^{(j)} \rvert.
\end{align}
Notice that by \eqref{eq:m<=0m<=0}, we have $\mu_{pp}^{(i)}\mu_{qq}^{(j)} \geq 0$.
If $\mu_{pp}^{(i)}\mu_{qq}^{(j)}=0$, then by \eqref{eq:muppqq>}, we have $\mu_{pq}^{(i)}\mu_{qp}^{(j)}=0$.
Then, the left-hand side of \eqref{eq:mm-mm>0} is equal to zero, which gives a contradiction.
So, we have
\begin{align}\label{eq:mm>0}
\mu_{pp}^{(i)}\mu_{qq}^{(j)}>0.
\end{align}
By \eqref{eq:m<=0m<=0} and \eqref{eq:mm>0}, we have 
\begin{align}\label{eq:m<0m<0}
\mu_{pp}^{(i)} < 0 \; \text{and} \; \mu_{qq}^{(i)} < 0.
\end{align}
By \eqref{eq:l>0l>0} and \eqref{eq:m<0m<0}, we have 
\begin{align}\label{eq:lm<0lm<0}
\lambda_{i^{\prime}}^*\mu_{pp}^{(i^{\prime})}h_p^*<0 \; \text{and} \; \lambda_{j^{\prime}}^*\mu_{qq}^{(j^{\prime})}h_q^*<0.
\end{align}
By \eqref{eq:jhlkl}, \eqref{eq:m<=0m<=0}, and \eqref{eq:lm<0lm<0}, we have 
\begin{align*}
&J(h^{*},\lambda^*)_{pp} = \sum_{i=1}^{t} \lambda_{i}^* \mu_{pp}^{(i)} h_p^* \leq \lambda_{i^{\prime}}^* \mu_{pp}^{(i^{\prime})} h_p^* <0, \; \text{and} \notag \\
&J(h^{*},\lambda^*)_{qq} = \sum_{i=1}^{t} \lambda_{i}^* \mu_{qq}^{(i)} h_q^* \leq \lambda_{j^{\prime}}^* \mu_{qq}^{(j^{\prime})} h_q^*<0.
\end{align*}
\end{proof}

\begin{lemma}\label{lemma:px>0}
 Suppose $x=(x_1, \ldots, x_n)$.
 Let $p(x)$ and $q(x)$ be polynomials in 
 $\mathbb{Q}[x]$  such that $p(x)-q(x)$ is either a zero polynomial or a sum of terms with positive coefficients. If there exists $x^* \in \mathbb{R}^n_{\geq 0}$ such that $q(x^*)>0$, then we have $p(x^*)>0$.
\end{lemma}
\begin{proof}
By the fact that $p(x)-q(x)$ is either a zero polynomial or a sum of terms with positive coefficients, we have $p(x^*)-q(x^*) \geq 0$. Therefore, we have $p(x^*)=(p(x^*)-q(x^*))+q(x^*)>0$.
\end{proof}

\begin{lemma}[Theorem 7.1.2 in \cite{mirsky2012}]\label{prop:akhl}
Given a matrix $M \in \mathbb{R}^{n \times n}$, let $p(z) := \operatorname{det}(z I-M) = z^n+a_1z^{n-1}+\cdots+a_n$. Then, for any $k \in \{1,\ldots,n\}$, we have
\begin{align*}
a_{k} = (-1)^{k} \sum_{1 \leq i_{1}<\ldots<i_{k} \leq n} \operatorname{det} \left( M[i_{1}, \ldots, i_{k}] \right).
\end{align*}
\end{lemma}

\begin{lemma}\label{lemma:positivedeth}
Let $G$ be a zero-one network. 
Let $J(h,\lambda) \in \Q[h,\lambda]^{s \times s}$ be the matrix corresponding to $G$ defined in \eqref{eq:defjhl}. 
Let $q(h,\lambda;z)$ be the polynomial defined in \eqref{def:qhlz}.
Let $H_{2}(h,\lambda)$ be the second Hurwitz matrix of $q(h,\lambda;z)$.
Then we have the following statements.
\begin{enumerate}[(i)]
\item For any $h^* \in \mathbb{R}_{> 0}^s$ and for any $\lambda^* \in \mathbb{R}_{\geq 0}^t$, if there exist $p$ and $q\;(1 \leq p<q \leq s)$ such that 
\begin{align*}
\operatorname{det} \left( J(h^*,\lambda^*)[p, q] \right)>0,
\end{align*}
then $\operatorname{det}(H_{2}(h^*,\lambda^*))>0$.
\item For any $h^* \in \mathbb{R}_{> 0}^s$ and for any $\lambda^* \in \mathbb{R}_{\geq 0}^t$, if there exist $p$, $q$, and $r\;(1 \leq p<q<r \leq s)$ such that 
\begin{align*}
d_{p,q,r}(h^*,\lambda^*)>0,
\end{align*}
then $\operatorname{det}(H_{2}(h^*,\lambda^*))>0$.
\end{enumerate}
\end{lemma}
\begin{proof}
By Lemma \ref{prop:akhl}, we have
\begin{align}\label{eq:a0hl}
b_{1}(h,\lambda) &= -\sum_{1 \leq i \leq s} \operatorname{det} \left( J[i] \right), \notag \\
b_{2}(h,\lambda) &= \sum_{1 \leq i<j \leq s} \operatorname{det} \left( J[i,j] \right), \text{ and} \notag \\
b_{3}(h,\lambda) &= -\sum_{1 \leq i<j<k \leq s} \operatorname{det} \left( J[i,j,k] \right).
\end{align}
By \eqref{def:dijk}, \eqref{eq:a0hl}, and by Definition \ref{def:hurwitz}, we have
{\small \begin{flalign*}
&\quad\quad\quad\quad\;\;\operatorname{det}(H_{2}(h,\lambda))&\notag
\end{flalign*}}
{\footnotesize\begin{flalign*}
\quad\quad\quad\quad=&\left | \begin{matrix}
b_{1}(h,\lambda) & 1  \\
b_{3}(h,\lambda) & b_{2}(h,\lambda)  \\
\end{matrix} \right |& \notag \\
\quad\quad\quad\quad=& \sum_{1 \leq i<j<k \leq s} \operatorname{det} \left( J[i, j, k] \right) - \sum_{1 \leq i \leq s} \operatorname{det} \left( J[i] \right) \sum_{1 \leq i<j \leq s} \operatorname{det} \left( J[i, j] \right) &
\end{flalign*}
\begin{flalign*}
\quad\quad\quad\quad=& \sum_{1 \leq i<j<k \leq s} \operatorname{det} \left( J[i, j, k] \right) - \sum_{1 \leq i<j \leq s} \sum_{1 \leq k \leq s \atop k \neq i,j} J_{kk} \operatorname{det} \left( J[i, j] \right) - \notag& \\
\quad\quad\quad\quad& \sum_{1 \leq i<j \leq s} (J_{ii}+J_{jj}) \operatorname{det} \left( J[i, j] \right)&
\end{flalign*}
\begin{flalign*}
\quad\quad\quad\quad=& \sum_{1 \leq i<j<k \leq s} \operatorname{det} \left( J[i, j, k] \right) - \sum_{1 \leq i<j \leq s \atop 1 \leq k < i} J_{kk} \operatorname{det} \left( J[i, j] \right) - \sum_{1 \leq i<j \leq s \atop i < k < j} J_{kk} \operatorname{det} \left( J[i, j] \right) - \notag &\\
\quad\quad\quad\quad& \sum_{1 \leq i<j \leq s \atop j < k \leq s} J_{kk} \operatorname{det} \left( J[i, j] \right) - \sum_{1 \leq i<j \leq s} (J_{ii}+J_{jj}) \operatorname{det} \left( J[i, j] \right)&
\end{flalign*}
\begin{flalign*}
\quad\quad\quad\quad=& \sum_{1 \leq i<j<k \leq s} \Big( \operatorname{det} \left( J[i, j, k] \right) - J_{ii} \operatorname{det} \left( J[j, k] \right) - J_{jj} \operatorname{det} \left( J[i, k] \right) - \notag &\\
\quad\quad\quad\quad& \quad\quad\quad\quad\quad\quad J_{kk} \operatorname{det} \left( J[i, j] \right) \Big) - \sum_{1 \leq i<j \leq s} (J_{ii}+J_{jj}) \operatorname{det} \left( J[i, j] \right) &
\end{flalign*}
\begin{flalign*}
\quad\quad\quad\quad=& \sum_{1 \leq i<j<k \leq s} (J_{ij}J_{jk}J_{ki} + J_{ik}J_{ji}J_{kj}-2J_{ii}J_{jj}J_{kk}) - \sum_{1 \leq i<j \leq s} (J_{ii}+J_{jj}) \operatorname{det} \left( J[i, j] \right)&
\end{flalign*}
\begin{flalign}\label{eq:deth2hl=}
\quad=& \sum_{1 \leq i<j<k \leq s} d_{i, j, k} - \sum_{1 \leq i<j \leq s} (J_{ii}+J_{jj}) \operatorname{det} \left( J[i, j] \right).&
\end{flalign}
} 
(i) By \eqref{eq:deth2hl=}, we have
{\small \begin{align}\label{eq:deth2+jppjqqdetjpq}
\operatorname{det}(H_{2}(h,\lambda))+(J_{pp}+J_{qq}) \operatorname{det} \left( J[p, q] \right) = \sum_{1 \leq i<j<k \leq s} d_{i, j, k} - \sum_{1 \leq i<j \leq s \atop i,j \ne p,q}(J_{ii}+J_{jj}) \operatorname{det} \left( J[i, j] \right).
\end{align}
}By Corollary \ref{coro:jkkis-} and by Lemma \ref{lemma:positivePM}, the right-hand side of \eqref{eq:deth2+jppjqqdetjpq} is either a zero polynomial or a sum of terms with positive coefficients.
By the hypothesis that $\operatorname{det} \left( J(h^*,\lambda^*)[p, q] \right)>0$ and by Lemma \ref{lemma:jpp<0jqq<0}, we have $J(h^*,\lambda^*)_{pp} <0$ and $J(h^*,\lambda^*)_{qq} <0$.
So, we have
\begin{align}\label{eq:-jpp+jqqdetjpq>0}
-(J(h^{*},\lambda^*)_{pp}+J(h^{*},\lambda^*)_{qq}) \operatorname{det} \left( J(h^*,\lambda^*)[p, q] \right)>0.
\end{align}
Then, by \eqref{eq:deth2+jppjqqdetjpq}, \eqref{eq:-jpp+jqqdetjpq>0}, and by Lemma \ref{lemma:px>0}, we have $\operatorname{det}(H_{2}(h^*,\lambda^*))>0$.

(ii) By \eqref{eq:deth2hl=}, we have
\begin{align}\label{eq:deth2-dpqr}
\operatorname{det}(H_{2}(h,\lambda))-d_{p,q,r} = \sum_{1 \leq i<j<k \leq s \atop i,j,k \ne p,q,r} d_{i, j, k} - \sum_{1 \leq i<j \leq s}(J_{ii}+J_{jj}) \operatorname{det} \left( J[i, j] \right).
\end{align}
By Corollary \ref{coro:jkkis-} and by Lemma \ref{lemma:positivePM}, the right-hand side of \eqref{eq:deth2-dpqr} is either a zero polynomial or a sum of terms with positive coefficients.
Then, by the hypothesis that $d_{p,q,r}(h^*,\lambda^*)>0$ and by Lemma \ref{lemma:px>0}, we have $\operatorname{det}(H_{2}(h^*,\lambda^*))>0$.
\end{proof}

\begin{lemma}\label{lemma:h1>0h2>0}
Let $G$ be a zero-one network with a stoichiometric matrix $\mathcal{N}$.
Let $r={\rm rank}({\mathcal N})$. 
Let $J(h,\lambda) \in \Q[h,\lambda]^{s \times s}$ be the matrix corresponding to $G$ defined in \eqref{eq:defjhl}. 
Let $q(h,\lambda;z)$ be the polynomial defined in \eqref{def:qhlz} and we write
\begin{align*}
q(h,\lambda;z)=\frac{1}{z^{s-r}}\operatorname{det}(z I-J(h,\lambda))=z^{r}+b_{1}(h,\lambda) z^{r-1}+\cdots+b_{r}(h,\lambda).
\end{align*}
For any positive integer $i\;(i\leq r)$, let $H_i(h,\lambda)$ be the $i$-th Hurwitz matrix of $q(h,\lambda;z)$.
Then for any $h^* \in \mathbb{R}_{> 0}^s$ and for any $\lambda^* \in \mathbb{R}_{\geq 0}^t$, we have 
\begin{enumerate}[(i)]
\item $\operatorname{det}(H_{1}(h^*,\lambda^*)) > 0$ if $b_{2}(h^*,\lambda^*) \neq 0$, and
\item $\operatorname{det}(H_{2}(h^*,\lambda^*)) > 0$ if $b_{3}(h^*,\lambda^*) \neq 0$.
\end{enumerate}
\end{lemma}
\begin{proof}
(i) By Lemma \ref{prop:akhl}, we have
\begin{align*}
b_{2}(h,\lambda) = \sum_{1 \leq i<j \leq s} \operatorname{det} \left( J[i, j] \right).
\end{align*}
By Lemma \ref{lemma:positivePM} (i), for any $1 \leq i<j \leq s$, $\operatorname{det} \left( J[i, j] \right)$ is either a zero polynomial or a sum of terms with positive coefficients.
So, for any $h^* \in \mathbb{R}_{> 0}^s$ and for any $\lambda^* \in \mathbb{R}_{\geq 0}^t$, $b_{2}(h^*,\lambda^*)\geq 0$.
If $b_{2}(h^*,\lambda^*) \neq 0$, then there exist $p$ and $q$ ($1 \leq p<q \leq s$) such that
\begin{align*}
\operatorname{det} \left( J(h^{*},\lambda^*)[p, q] \right)> 0
\end{align*}
So, by Lemma \ref{lemma:jpp<0jqq<0}, we have 
\begin{align}\label{eq:1deth1>0}
J(h^*,\lambda^*)_{pp} <0 \text{ and }J(h^*,\lambda^*)_{qq} <0.
\end{align}
By Definition \ref{def:hurwitz} and by Lemma \ref{prop:akhl}, we have
\begin{align}\label{eq:2deth1>0}
\operatorname{det}(H_{1}(h,\lambda)) = b_{1}(h,\lambda) = -\sum_{1 \leq i \leq s} J_{ii}.
\end{align}
By Corollary \ref{coro:jkkis-}, we have
\begin{align}\label{eq:3deth1>0}
J(h^*,\lambda^*)_{ii} \leq 0, \text{ for any } i \in \{1,\ldots,s\}.
\end{align}
By \eqref{eq:1deth1>0}, \eqref{eq:2deth1>0}, and \eqref{eq:3deth1>0}, we have
\begin{align*}
\operatorname{det}(H_{1}(h^*,\lambda^*)) \geq -J(h^*,\lambda^*)_{pp} >0.
\end{align*}

(ii) For any $h^* \in \mathbb{R}_{> 0}^s$ and $\lambda^* \in \mathbb{R}_{\geq 0}^t$ such that 
\begin{align*}
b_{3}(h^*,\lambda^*)=-\sum_{1 \leq i<j<k \leq s} \operatorname{det} \left( J(h^*,\lambda^*)[i, j, k] \right) \neq 0,
\end{align*}
there exist $p$, $q$, and $r$ ($1 \leq p<q<r \leq s$) such that
\begin{align}\label{eq:detjpqrhl>0}
\operatorname{det} \left( J(h^*,\lambda^*)[p, q, r] \right) \neq 0.
\end{align}
By Lemma \ref{lemma:positivePM} (i), we have
\begin{align}\label{eq:3det}
\operatorname{det} \left( J(h^*,\lambda^*)[p, q] \right) \geq 0, \;\;\operatorname{det} \left( J(h^*,\lambda^*)[p, r] \right) \geq 0, \;\;\text{and } \operatorname{det} \left( J(h^*,\lambda^*)[q, r] \right) \geq 0.
\end{align}
If there exists one inequality in \eqref{eq:3det} to be a strict inequality, without loss of generality, we assume $\operatorname{det} \left( J(h^*,\lambda^*)[p, q] \right) > 0$, then by Lemma \ref{lemma:positivedeth} (i), we have
$\operatorname{det}(H_{2}(h^*,\lambda^*))>0$.
Recall the definition of $d_{i,j,k}$ in \eqref{def:dijk}.
We have
\begin{align}\label{eq:dpqr=}
d_{p,q,r} =& J_{pq}J_{qr}J_{rp} + J_{pr}J_{qp}J_{rq}-2J_{pp}J_{qq}J_{rr} \notag \\
=& J_{pp}J_{qq}J_{rr} + J_{qp}J_{rq}J_{pr} + J_{rp}J_{pq}J_{qr} - J_{pp}J_{qr}J_{rq} - J_{pr}J_{qq}J_{rp} - J_{pq}J_{qp}J_{rr} \notag \\
& - J_{pp}(J_{qq}J_{rr}-J_{qr}J_{rq}) - J_{qq}(J_{pp}J_{rr}-J_{pr}J_{rp}) - J_{rr}(J_{pp}J_{qq}-J_{pq}J_{qp}) \notag \\
=& \operatorname{det} \left( J[p,q,r] \right) - J_{pp} \operatorname{det} \left( J[q,r] \right) - J_{qq} \operatorname{det} \left( J[p,r] \right) - J_{rr} \operatorname{det} \left( J[p,q] \right).
\end{align}
If all the three inequalities in \eqref{eq:3det} are equal to 0, then by \eqref{eq:detjpqrhl>0} and \eqref{eq:dpqr=}, $d_{p,q,r}(h^*,\lambda^*) \ne 0$.
By Lemma \ref{lemma:positivePM} (ii), $d_{p,q,r}(h^*,\lambda^*) \ge 0$.
So, $d_{p,q,r}(h^*,\lambda^*) > 0$.
Then, by Lemma \ref{lemma:positivedeth} (ii), we have $\operatorname{det}(H_{2}(h^*,\lambda^*))>0$.
\end{proof}

\begin{proof}[\bf{Proof of Theorem \ref{thm:rank<3}}]
By the definition of network (see \eqref{eq:network}), $\mathcal{N}$ is not a zero matrix. So, we have ${\rm rank}({\mathcal N}) \geq 1$.
If ${\rm rank}({\mathcal N}) = 1$, then by \eqref{def:reduceJac} and \eqref{eq:jacfkx}, for any $\kappa^* \in \R^m_{>0}$ and for any $x^* \in \R^m_{\geq0}$, the rank of ${\rm Jac}^{\text{red}}_f(\kappa^*,x^*)$ is no more than 1.
In this case, ${\rm Jac}^{\text{red}}_f(\kappa^*,x^*)$ does not have a pair of complex-conjugate eigenvalues.
So, by Definition \ref{def:kxHB}, the network $G$ does not admit a Hopf bifurcation. 

Let $J(h,\lambda)$ be the matrix corresponding to $G$ defined in \eqref{eq:defjhl}. 
Let $q(h,\lambda;z)$ be the polynomial defined in \eqref{def:qhlz}, and let $H_i(h,\lambda)$ be $i$-th Hurwitz matrix of $q(h,\lambda;z)$. 
If ${\rm rank}({\mathcal N})=2$, then we can rewrite $q(h,\lambda;z)$ as
\begin{align*}
q(h,\lambda;z) = z^{2} + b_{1}(h,\lambda) z + b_{2}(h,\lambda).
\end{align*}
By Lemma \ref{lemma:h1>0h2>0} (i), for any $h^* \in \mathbb{R}^s_{> 0}$ and for any $\lambda^* \in \mathbb{R}^s_{\geq 0}$, $\operatorname{det}(H_1(h^*,\lambda^*))>0$ if $b_{2}(h^*,\lambda^*) \neq 0$. Therefore, by Lemma \ref{lemma:preHB}, the network $G$ does not admit a Hopf bifurcation.
If ${\rm rank}({\mathcal N})=3$, then we can rewrite $q(h,\lambda;z)$ as
\begin{align*}
q(h,\lambda;z) = z^{3} + b_{1}(h,\lambda) z^{2} + b_{2}(h,\lambda) z + b_{3}(h,\lambda).
\end{align*}
By Lemma \ref{lemma:h1>0h2>0} (ii), for any $h^* \in \mathbb{R}^s_{> 0}$ and for any $\lambda^* \in \mathbb{R}^s_{\geq 0}$, $\operatorname{det}(H_2(h^*,\lambda^*))>0$ if $b_{3}(h^*,\lambda^*) \neq 0$. Therefore, by Lemma \ref{lemma:preHB}, the network $G$ does not admit a Hopf bifurcation.
To sum up, if ${\rm rank}({\mathcal N}) \leq 3$, then the network $G$ does not admit a Hopf bifurcation.
\end{proof}
\section{Discussion}\label{sec:dis}
Recall that Question \ref{ques:small} has been studied for the bimolecular networks in \cite{biHB2022}, and in this paper, we study Question \ref{ques:small} for the zero-one networks. So, a natural direction is to study the question for more general networks. For instance, one can ask when the number of reactants is given, which small networks admit Hopf bifurcations/oscillations.   








\end{document}